\title{%
	\vspace*{-4ex}%
	On exact counting and quasi-quantum complexity%
}
\date{25 September 2015 \vspace*{-2ex}}
\author{Niel de Beaudrap\thanks{niel.debeaudrap@gmail.com} \\ Department of Computer Science, University of Oxford}
\pdfoutput=1

\documentclass[a4paper,10pt]{article}
\usepackage[utf8x]{inputenc}
\usepackage{microtype} 
\usepackage{times}
\usepackage{enumitem}
\usepackage[font=small]{caption}

\usepackage{amsmath,amssymb,amsthm,bbm,mathtools}
\usepackage{enumitem}
\usepackage[font={small}]{caption}
\usepackage{euscript}
\usepackage{hyperref}

\hypersetup{
	pdftitle={On exact counting and quasi-quantum complexity},
	pdfauthor={Niel de Beaudrap},
	colorlinks,%
	citecolor=blue!50!black,%
	filecolor=black,%
	linkcolor=black,%
	urlcolor=black
}

\usepackage{tikz}
\usetikzlibrary{shapes,arrows,calc,positioning,fit,shadows}

\pgfdeclarelayer{background}
\pgfsetlayers{background,main}

\newcommand\ie{\emph{i.e.}}
\newcommand\eg{\emph{e.g.}}

\newcommand\etal{\emph{et al.}}

\newcommand\YES{\emph{yes}}
\newcommand\NO{\emph{no}}

\theoremstyle{definition}
\newtheorem{definition}{Definition}

\newtheorem*{definition*}{Definition}

\theoremstyle{plain}
\newtheorem{theorem}{Theorem}
\newtheorem{lemma}[theorem]{Lemma}
\newtheorem{proposition}[theorem]{Proposition}
\newtheorem*{proposition*}{Proposition}

\newtheorem*{claim*}{Claim}
\newtheorem{corollary}{Corollary}[theorem]

\DeclareMathOperator\poly{poly}
\DeclareMathOperator\diag{diag}
\DeclareMathOperator\GL{GL}

\let\P\relax
\newcommand\DeclareComplexityClass[2]{\newcommand#1{\ensuremath{\mathsf{#2}}}}
\DeclareComplexityClass\co{co}
\DeclareComplexityClass\Ceq{C_{\texttt=}}
\DeclareComplexityClass\Mod{Mod}
\DeclareComplexityClass\Post{Post}
\DeclareComplexityClass\P{P}
\DeclareComplexityClass\FP{FP}
\DeclareComplexityClass\mU{U}
\DeclareComplexityClass\UP{UP}
\DeclareComplexityClass\mGL{GL}
\DeclareComplexityClass\BPP{BPP}
\DeclareComplexityClass\WPP{WPP}
\DeclareComplexityClass\LWPP{LWPP}
\DeclareComplexityClass\LPWPP{LPWPP}
\DeclareComplexityClass\SPP{SPP}
\DeclareComplexityClass\NP{NP}
\DeclareComplexityClass\BQP{BQP}
\DeclareComplexityClass\NQP{NQP}
\DeclareComplexityClass\QMA{QMA}
\DeclareComplexityClass\EQP{EQP}
\DeclareComplexityClass\ZQP{ZQP}
\DeclareComplexityClass\PP{PP}
\DeclareComplexityClass\AC{AC}
\DeclareComplexityClass\NC{NC}
\DeclareComplexityClass\GapP{GapP}
\DeclareComplexityClass\AWPP{AWPP}

\newcommand\GLP[1][]{\ensuremath{\mathsf{GLP}_{#1}}}
\newcommand\UnitaryP[1][]{\ensuremath{\mathsf{UnitaryP}_{\!#1}}}
\newcommand\AffineP[1][]{\ensuremath{\mathsf{AffineP}_{\!#1}}}

\newcommand\C{\mathbb{C}}

\newcommand\N{\mathbb{N}}
\newcommand\Z{\mathbb{Z}}

\renewcommand\vec\mathbf

\newcommand\herm{^\dagger}

\newcommand\ox{\otimes}
\newcommand\x{\times}
\newcommand\sox[1]{^{\otimes #1}}
\newcommand\sur[1]{^{(#1)}}

\let\subset\subseteq
\let\oldepsilon\epsilon
\let\epsilon\varepsilon
\let\varepsilon\oldepsilon
\let\le\leqslant
\let\ge\geqslant

\newcommand\idop{\mathbbm{1}}

\makeatletter
	\newcommand\ket[1]{\left| #1 \right\rangle\@ifnextchar\bra{\mspace{-4mu}}{}}
	\newcommand\bra[1]{\left\langle #1 \right|}

	\newcommand\bracket[2]{\left\langle #1 \left| #2 \right\rangle\right.}
\makeatother

\begin{document}

\maketitle

\begin{abstract}
	We present characterisations of ``exact'' gap-definable classes, in terms of indeterministic models of computation which slightly modify the standard model of quantum computation.
	This follows on work of Aaronson~\cite{Aaronson-2005}, who shows that the counting class $\PP$ can be characterised in terms of bounded-error ``quantum'' algorithms which use invertible (and possibly non-unitary) transformations, or postselections on events of non-zero probability.
	Our work considers similar modifications of the quantum computational model, but in the setting of exact algorithms, and algorithms with zero error and constant success probability.
	We show that the gap-definable~\cite{FFK-1994} counting classes which bound exact and zero-error quantum algorithms can be characterised in terms of ``quantum-like'' algorithms involving nonunitary gates, and that postselection and nonunitarity have equivalent power for exact quantum computation only if these classes collapse.
 \end{abstract}

\section{Introduction}

The relationship between quantum computation (as captured by the class \BQP), and complexity classes involving classical nondeterminism, is still unclear.
It is known that $\BPP \subset \BQP$; is this containment strict?
Furthermore, it is not known whether $\NP \subset \BQP$ or $\BQP \subset \NP$ hold, and it is conjectured (see for example Refs.~\cite{AA-2009,AA-2011}) that neither containment holds: can this be shown?
If $\NP$ and $\BQP$ are indeed incomparable, is there any natural relationship between quantum computation and nondeterministic Turing machines?
These problems are expected to be very difficult.
In fact, the best known lower bounds on quantum complexity classes are the classes $\P$ and $\BPP$, and the best known upper bounds are \emph{gap-definable classes}~\cite{FFK-1994} (described below), most of which also do not have a well-understood relationship to $\NP$.

To look for complexity-theoretic lower bounds for \BQP, one may consider exact or zero-error quantum algorithms.
An \emph{exact} quantum algorithm is one in which the output bit is in one of the pure states $\ket{0}$ or $\ket{1}$, according to whether the input is a \NO\ or a \YES\ instance.
For zero-error quantum algorithms, we allow a bounded probability that the algorithm fails to decide between \YES\ and \NO\ instances (indicated by a measurement on some non-output qubit), but require otherwise that the algorithm indicates the correct answer.
The sets of problems which can be decided by such algorithms (represented by polynomial-time uniform circuit families), using finite gate sets, are the classes $\EQP$~\cite{BV-1997} and $\ZQP$~\cite{Nishino-2002,BW-2003} respectively.
The class \BQP\ is the class of problems which can be decided by such circuits, with bounded error; we then have $\EQP \subset \ZQP \subset \BQP$.

One might hope that the exactness constraint might yield an elegant theory of such algorithms, as is the case for other problems in quantum informatics involving one-sided error, such as the zero-error classical capacity of quantum channels~\cite{MdA-2005,CLMW-2010,Beigi-2010}, and 2-QUANTUM-SAT~\cite{Bravyi-2006,CCDJZ-2011,dBOE-2010} (in contrast to 2-LOCAL-HAMILTONIAN~\cite{KKR-2006}).
However, the study of $\EQP$ seems to be hampered by the restriction to finite gate-sets, which makes it difficult to produce exact algorithms for natural problems, except for some in the oracle model (\eg~\cite{Ambainis-2013,AIS-2013}).
One may attempt to overcome this difficulty by considering a more general model of computation in which we only require the gate-coefficients to be efficiently computable as part of the uniformity condition of the circuit family, and require exact decision. (Some of our results consider such models.)
Nevertheless, as there seems to be renewed interest in exact quantum query algorithms~\cite{MJM-2015}, one may hope that the theory of exact and zero-error algorithms with finite gate-sets may yet see some progress in the near future.

Each of $\EQP$, $\ZQP$, and $\BQP$ are bounded above by \emph{gap-definable classes}~\cite{FFK-1994}, which distinguish between \YES\ and \NO\ instances via differences between the number of accepting/rejecting branches of nondeterministic Turing machines (NTMs).
Figure~\ref{fig:containments} presents the best previously upper known bounds for these quantum classes, which are $\EQP \!\!\:\subset\!\!\: \LWPP$,\; $\ZQP \!\!\:\subset\!\!\: {\Ceq\P \!\!\;\cap\!\!\; \co\Ceq\P}$,\; and $\BQP \!\!\:\subset\!\!\: \AWPP$, which follow from Refs.~\cite{FGHP-1999,FR-1999}.\footnote{%
	The bound on $\ZQP$ follows from $\ZQP \subset \NQP \cap \co\NQP$ (Proposition~\ref{prop:ZQPupperBound} on page~\pageref{prop:ZQPupperBound}).
	A problem is in $\NQP$ if there are quantum circuits which yield an output of $1$ with non-zero probability precisely for the \YES\ instances (in analogy to a probabilistic formulation of $\NP$).
	Fenner~\etal~\cite{FGHP-1999} show that $\NQP = \co\Ceq\P$.
}
\begin{figure}
\begin{minipage}{0.5\textwidth}
\small
\begin{tikzpicture}[yscale=0.75, xscale=1.4, every node/.style={draw=none, inner sep=0pt, outer sep=3pt, rectangle}]
	\node	(P) at (0,0) {\P};
	\node (EQP) at (2,2) {\EQP};
  \node (ZQP) at (2,5) {\ZQP};
  \node (BQP) at (2,7) {\BQP};
  \node	(UP)	at (0,2) {\UP};
  \node (SPP)	at (0,3) {\SPP};
  \node (LWPP) at (0,6) {\LWPP};
  \node (DCeqP) at (0,9) {$\Delta \Ceq\P$};
  \node (WPP) at (0,7) {\WPP};
  \node (coCeqP) at (0,11) {$\co\Ceq\P$};
  \node [anchor=north east] at (coCeqP.north west) {$\NQP = {}\!\!\!\!\!\:$};
  \node (AWPP) at (1,9) {\AWPP};
	\node (PP) at (1,11) {$\PP \mathrlap{{}=\BQP_{\GL} \!\!\:= \Post\BQP}$};
	\draw (EQP) -- (ZQP) -- (BQP);
	\draw (P) -- (UP) -- (SPP) -- (LWPP) -- (WPP) -- (DCeqP) -- (coCeqP); 
	\draw [shorten >=-4pt] (P) -- (EQP);
	\draw (EQP) -- (LWPP);
	\draw (WPP) -- (AWPP);
	\draw (ZQP) -- (DCeqP);
	\draw (ZQP) -- (BQP);
	\draw (DCeqP) -- (PP);
	\draw (BQP) -- (AWPP);
	\draw (AWPP) -- (PP);
	\begin{pgfonlayer}{background}
		\node [fill=gray!20!white, fit=(coCeqP)(SPP)(AWPP)(PP), inner xsep=15pt, inner ysep=-26pt] {};
	\end{pgfonlayer}
\end{tikzpicture}
~
\end{minipage}
\hspace*{-0.1\textwidth}
\hfill
\begin{minipage}{0.54\textwidth}
\vspace*{11ex}
\caption{%
	The previously known relations between gap-definable classes~\cite{FFK-1994} (contained in the grey box), and quantum computational complexity classes.
	The left-most column consists of exact counting classes, including \UP\ (problems decided by NTMs according to whether the number of accepting branches is zero or one) and \P\ (which further requires that there be a total of one branch) for context.
	The classes $\SPP \subset \LWPP \subset \WPP \subset \co\Ceq\P$ distinguish between \NO\ and \YES\ instances $x \in \{0,1\}^\ast$ according to whether the acceptance gap is zero or positive.
	For \WPP, if the gap is positive, we require that it equal some $h(x) > 0$ which is computable from $x$ in deterministic time $O(\poly |x|)$.
	For \LWPP, we further require that $h(x)$ depend only on $|x|$, and for \SPP\ we require that $h$ be constant.
	Here, $\Delta\Ceq\P := {\Ceq\P \cap \co\Ceq\P}$.
	The equalities are results of Fenner~\etal~\cite{FGHP-1999} and Aaronson~\cite{Aaronson-2005}, and the containments $\EQP \subset \LWPP$ and $\BQP \subset \AWPP$ are proven by Fortnow and Rogers~\cite{FR-1999}.
}
\label{fig:containments}
\end{minipage}
\vspace*{-2.75ex}
\end{figure}
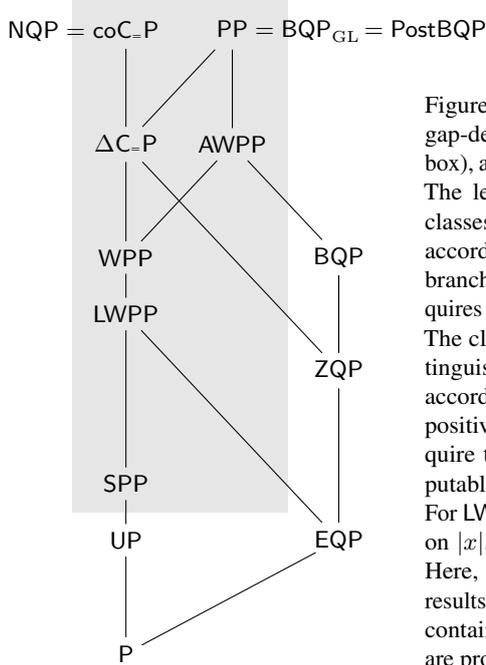
For the class \Ceq\P, we require there be an NTM for which this gap is precisely zero for \YES\ instances, and strictly positive for \NO\ instances.
For $\Delta\Ceq\P := {\Ceq\P \cap \co\Ceq\P}$, we require that there \emph{also} be an NTM with the gap conditions reversed: these represent a pair of zero-error randomised algorithms in which exactly one has any bias towards accepting a given input.
(Whichever algorithm has a bias, indicates whether the input is a \YES\ or \NO\ instance.)
For \LWPP, we further require that these NTMs have an acceptance gap which is either zero, or an efficiently computable function of the length of the input; the class \AWPP\ is a sort of bounded-error version of \LWPP~\cite{Fenner-2003}.

These gap-definable classes may seem quite technical on first encounter.
Still, given that \BPP\ can be characterised in terms of the number of just the accepting branches of NTMs,
this approach to relate quantum computation to nondeterministic complexity may seem promising if one considers the cancellation of amplitudes (\ie~destructive interference) to be the key distinction between quantum computation and randomised computation.
However, as $\UP$ is a subset of each of these gap-definable classes, we would have $\NP \subset \BQP$ by Valiant--Vazirani~\cite{VV-1986} if any of the quantum classes were equal to such a gap-definable class.
Thus we might expect these bounds to be loose.

How might we explore the relationship between gap-functions and quantum computation, to find better upper bounds on the quantum complexity classes?
One approach is to modify the standard quantum computational model, characterise the power of the modified model, and consider the role played by the modifications in the characterisation.
For instance, if one substitutes the real or complex numbers in quantum operations by elements of the ring $\Z_k$ for $k$ a prime power, one obtains new models of computation~\cite{Beaudrap-2014} in which exact or bounded-error computation both efficiently decide any $L \in \Mod_k\P$~\cite{BGH-1990}, classes which again contain $\UP$.
This suggests that the complex amplitudes in quantum computation serve to \emph{constrain} the power of the computational model as much as to empower it, assuming that $\UP \not\subseteq \EQP$.
Another modification, considered by Aaronson~\cite{Aaronson-2005}, is bounded-error ``quantum'' algorithms involving invertible (possibly non-unitary) transformations.
He shows that for bounded-error computation, this is equivalent to allowing the algorithm to \emph{postselect} (condition its output on measurement events of non-zero probability), and suffices to decide problems in $\PP$ with bounded error.
This indicates the computational power of postselection, and suggests that any attempt to separate $\BQP$ from $\PP$ must somehow account for the unitarity of the operations in quantum algorithms.

\vspace*{-1ex}
\paragraph{Results.}

Extending Aaronson~\cite{Aaronson-2005}, we characterise the power of invertible transformations and of postselection for exact and zero-error quantum algorithms.
We first describe a potentially tighter bound $\EQP \subset \LPWPP$ by accounting for the finite gate set of \EQP-type algorithms.
(This again suggests that finite gate-sets in the zero-error setting may be a strong restriction.)
We then characterise \LPWPP\ and \LWPP\ in terms of exact ``quantum-like'' algorithms, using invertible gates 
(from an infinite but polynomial-time specifiable~\mbox{\cite[\S3.3\,A]{Beaudrap-2014}} gate-set in the case of \LWPP).
We also show that $\Delta\Ceq\P$ can be characterised either through zero-error quantum-like algorithms using invertible gates, or exact quantum algorithms using postselection.
(Thus, postselection and non-unitary operations have equivalent power for exact algorithms only if $\LWPP = \Delta\Ceq\P$.)

Our results demonstrate that the exact gap-definable classes \LWPP\ and $\Delta\Ceq\P$ (and the new class \LPWPP) can be described simply in terms of \emph{quasi-quantum} computation, in which we allow invertible non-unitary gates.
At the same time, this shows that the exactness condition \emph{per se} does not represent a barrier to exact quantum algorithms, as these exact quantum-like classes contain problems of interest (such as \textrm{\small GRAPH ISOMORPHISM}~\cite{AK-2002}).
Our results also show that unitarity is a significant constraint on quantum algorithms even in the exact setting, where the Born rule plays no role.

\vspace*{-1ex}
\paragraph{Structure of the paper.}
Section~2 presents some preliminaries on quantum complexity and counting classes, defining all of the complexity classes which we use.
Section~3 contains a few technical results in gap-function complexity.
Section~4 presents our two characterisations of $\Delta\Ceq\P$, while Section~5 presents our results on exact quantum-like algorithms with invertible gates.
We conclude in Section~6 with some remarks about other gap-definable classes, and open questions.

\section{Preliminaries}
\label{sec:preliminaries}

\subsection{Definitions in counting complexity}

We begin with some basic terminology relating to the accepting and rejecting branches of poly-time nondeterministic Turing machines (NTMs).

We consider NTMs in \emph{normal form}~\cite{FFK-1994}, in which each nondeterministic transition selects from two possible transitions, and in which the number of nondeterministic transitions made in any computational branch on input $x \in \{0,1\}^\ast$ is $q(|x|)$, for some fixed polynomial $q$.
We may then represent a computational branch of the NTM by a boolean string $b \in \{0,1\}^{q(|x|)}$, which we call a \emph{branching string} of the NTM.

Probabilities can be modelled by NTMs by counting the number of branches (conceived of as arising from unbiased coin-flips) which terminate either in the ``accept'' or ``reject'' state.
To instead model amplitudes which can destructively interfere, we represent negative contributions to amplitudes by rejecting branches of an NTM, and positive contributions by accepting branches; the difference in the numbers of these represent the accumulated amplitude.
This motivates studying complexity classes defined in terms of \emph{gap functions}~\cite{FFK-1994}:

\begin{definition}
  \GapP\ is set of integer-valued functions $g$ on finite input-strings, for which there is a poly-time NTM $\mathbf N$, such that the difference between the number of accepting and rejecting branches of $\mathbf N$ is $g(x)$ for each input $x$.
  We call $g$ the \emph{gap function of $\mathbf N$}.
\end{definition}
\noindent

For $g$ the gap function of an NTM in normal form, $g(x)$ is always even; however, by Ref.~\cite[Lemma~4.3]{FFK-1994}, there will be a gap function $h \in \GapP$ of some other NTM not in normal form, such that $g(x) = 2h(x)$.
While we present our results for NTMs $\mathbf N$ in normal form, we take the liberty of considering ``half-gap functions'' $h(x) = \tfrac{1}{2}g(x)$, for $g$ the gap function of $\mathbf N$.

A \emph{counting class} is a class of languages which can be decided in terms of the number of accepting or rejecting branches of a poly-time NTM.
The class $\Ceq\P$ is sometimes referred to as the ``exact counting'' class, as it may be defined in terms of polytime NTMs, for which an input is a \YES\ instance if and only if the number of accepting branches is equal to a given polytime-computable function.
In this article, we call \emph{any} complexity class an exact counting class if it distinguishes between \YES\ and \NO\ instances according to whether the number of its accepting branches, or its gap function, equals some efficiently computable function.
Following Ref.~\cite{FFK-1994}, we may define the exact counting classes relevant to our results as follows:
\begin{definition} 
	We define the classes $\Ceq\P$, $\co\Ceq\P$, $\WPP$, $\LWPP$, and $\SPP$ as follows.
	\begin{itemize}
	\item 
		$\Ceq\P$ is the class of languages $L$ for which there is $g \in \GapP$, such that $x \in L$ if and only if $g(x) = 0$; $\co\Ceq\P$ is the class of problems for which there is $g \in \GapP$, such that $x \in L$ if and only if $g(x) \ne 0$.
		We denote $\Delta\Ceq\P := \Ceq\P \cap \co\Ceq\P$.
	\item
		$\WPP$ is the class of $L \in \Delta\Ceq\P$ for which, in addition to the above (for $g$ the gap function of an NTM in normal form), there exists a \emph{poly-time computable} integer function $h$, such that either $x \notin L$ and $g(x) = 0$ or $x \in L$ and $g(x) = 2 h(x) \ne 0$.
	\item
		$\LWPP$ is the class of $L \in \WPP$ for which, in addition to the above, we may require that $h(x)$ depend only on $|x|$; we then say that $h$ is \emph{length-dependent}.
	\item
		$\SPP$ is the class of $L \in \LWPP$ for which, furthermore, we may require $h(x) = 1$ (or more generally $h(x) \in O(\poly |x|)$ by~\cite[Theorem~5.9]{FFK-1994}).
	\end{itemize}
\end{definition}
Precisely evaluating gap-functions is difficult in general (for example, simple transformations of NTMs suffice to show that $\UP \subset \SPP$ and $\NP \subset \co\Ceq\P$~\cite{FFK-1994}), so exact counting classes such as these may be considered quite powerful.

\subsection{Upper bounds on quantum complexity}

While $\Ceq\P$ is a subject of some interest in counting complexity~\cite{BGH-1990,TO-1992,FFK-1994,FGHP-1999,STT-2005}, the classes \LWPP\ and \WPP\ appear to be of more technical importance~\cite{FR-1999,ST-2006}, largely for their relationships to quantum complexity.
In each case, the polytime computable function $h$ represents a sort of ``normalising factor'' for the gap-function $g$, and we distinguish between \NO\ and \YES\ instances according to whether $\tfrac{1}{2}g(x)/h(x) = 0$ or $\tfrac{1}{2}g(x)/h(x) = 1$.
This intuition motivates the following ``approximate counting'' class corresponding to \LWPP\ and \WPP~\cite{FFKL-1993}:
\begin{definition}
	\label{def:AWPP}
	For $\epsilon > 0$, let $a \approx_\epsilon b$ if and only if $|a - b| \le \epsilon$.
	Then \AWPP\ is the class of languages $L$ such that, for any $\epsilon \in 2^{-O(\poly n)}$, there is a gap-function $g \in \GapP$ and a poly-time computable length-dependent function $h$, such that for all inputs $x$ we have $0 \le g(x) \le h(x)$, and either $x$ is a \NO\ instance and $g(x)/h(x) \approx_{\epsilon} \!0$ or $x$ is a \YES\ instance and $g(x)/h(x) \approx_{\epsilon} \!1$.
	Furthermore, we obtain the same class if there are such gap-functions $g$ even if we restrict to $\epsilon = \tfrac{1}{3}$ and power functions $h(x) = 2^{t(|x|)}$~\cite{Fenner-2003}.
\end{definition}
\noindent
We consider quantum circuits whose gates involve only algebraic coefficients (without loss of generality~\cite{ADH-1997}), given exactly as rational combinations of products of independent algebraic numbers, thus admitting an efficient algorithm for deciding equality.
Representing amplitudes by ``normalised'' gap functions provides intuition for the following:
\begin{proposition}[{Fortnow and Rogers~\cite{FR-1999}}]
	\label{prop:classicUpperBounds}
	$\EQP \subset \LWPP$ and $\BQP \subset \AWPP$.  
\end{proposition}
\noindent
We may refine the upper bound on $\EQP$ by defining an intermediate class to $\SPP$ and $\LWPP$, in which the length-dependent function evaluates powers of some fixed integer, following the characterisation of \AWPP\ by Fenner~\cite{Fenner-2003} described in Definition~\ref{def:AWPP}:
\begin{definition} 
  $\LPWPP$ is the class of problems in $\LWPP$ for which there exists an integer $M \ge 1$, a poly-time computable length-dependent function $h$, and a gap-function $g \in \GapP$, such that for all inputs $x$ we have $h(x) = M^t$ for some $t \ge 0$; and either $x$ is a \NO\ instance and $g(x) = 0$, or $x$ is a \YES\ instance and $\tfrac{1}{2}g(x) = h(x)$.
\end{definition}
\noindent
We have $\SPP \subset \LPWPP$ by definition (take $M = 1$ or $t = 0$).
As is usual in complexity theory, it is not obvious whether either of the containments $\SPP \subseteq \LPWPP \subseteq \LWPP$ are strict.
It is quite plausible that $\LPWPP \ne \LWPP$; though perhaps the restriction in the definition of $\LPWPP$ of the half-gap functions $h$ to perfect powers may allow such problems to subsumed by $\SPP$.

\begin{proposition}
	\label{prop:LPWPPunitaryBound}
	$\EQP \subset \LPWPP$.
\end{proposition}
\begin{proof}
	This containment is implicit in Ref.~\cite[Theorem~3.8]{FR-1999}: the length-dependent poly-time computable function in their proof computes powers of some $M \ge 1$, where $M$ is the common denominator of unitary gate coefficients expressed in rationalised form (\ie~with positive integer denominators).
\end{proof}
For zero-error quantum computations, the best known bounds follow from the unbounded-error case.
Let $\NQP$ be the set of problems for which there is a uniform family $\{C_n\}_{n \ge 1}$ of unitary circuits, for which $C_n \ket{x} = \ket{\psi(x)}\ket{0}$ for some $\ket{\psi(x)}$ if $x$ is a \NO\ instance, and where $C_n \ket{x}$ does not factor in this way (\ie~yields the output $\ket{1}$ with non-zero probability) if $x$ is a \YES\ instance.
We may then show:
\begin{proposition}
	\label{prop:ZQPupperBound}
  $\ZQP \subset \Delta\Ceq\P$.
\end{proposition}
\begin{proof}
  Consider a $\ZQP$ algorithm for a problem, which succeeds with probability $p \in \Omega(1)$, indicated by the outcome $\ket{1}$ on the measurement of some qubit $s$.
	If $a$ is the output qubit, there is a non-zero probability that $(s,a)$ is in the state $\ket{10}$ for \NO\ instances, with no probability of this outcome for \YES\ instances; and a non-zero probability that $(s,a)$ is in the state $\ket{11}$ for \YES\ instances, with no probability of this outcome for \NO\ instances.
	We may use this to produce \NQP\ and \co\NQP\ algorithms for the problem, so that $\ZQP \subset \NQP \cap \co\NQP$.
	The proposition then follows from $\Delta\Ceq\P := \Ceq\P \cap \co\Ceq\P$ and the result of Fenner~\etal~\cite{FGHP-1999} that $\NQP = \co\Ceq\P$
\end{proof}

The class $\BQP$ is widely accepted in the literature as a quantum complexity class of interest.
Even the less-studied class $\ZQP$ contains problems not expected to be solvable with bounded-error by classical algorithms.
(As Nishimura and Ozawa~\cite{NO-2002} point out,\label{discn:factoring} using a zero-error primality testing algorithm and Shor's algorithm as subroutines, we may test whether an integer has a prime factor larger than some threshold $k$ with zero error and at least $\tfrac{1}{2}$ probability of success with quantum algorithms.)
However, there does not yet appear to be any problem known to be in $\EQP$ which is thought to lie outside of $\P$.\footnote{%
	The closest result known to the author is a circuit family for the discrete logarithm problem due to Mosca and Zalka~\cite{MZ-2003}: however, this result involves the preparation of input-dependent quantum superpositions, which cannot be realised in any obvious way using gates from a fixed finite set acting on standard basis states.}
One might suspect that the restrictions imposed on $\EQP$ (to problems that can be decided exactly, using circuit families $\{C_n\}_{n\ge 1}$ composed of unitary gates, drawn from a single finite basis) may be too restrictive to allow an interesting theory of non-classical algorithms.
If we suppose that unitarity is necessary for a physically-motivated model of computation, but we are still interested in principle in which problems could be decided exactly, we may consider families of circuits which are not constructed from a single finite basis.
We may consider circuits with ``potentially infinite'' gate-sets, using the framework of \mbox{Ref.~\cite[\S3.3\,A]{Beaudrap-2014}}:\footnote{%
	Similar families of circuits are those described simply as ``uniform'' by Nishimura and Ozawa~\cite{NO-2002}; however, while they suppose that coefficients are specified by rational approximations, we require an exact representation for which there exist efficient algorithms for arithmetic operations.
}

\begin{definition}
	\label{def:UnitaryP}
  \UnitaryP[\C] is the class of problems for which there is a circuit family $\{ C_n \}_{n \ge 1}$ with unitary gates and preparation of fresh qubits in the state $\ket{0}$, which
  \begin{itemize}
  \item
		is \emph{polynomial-time uniform}: the structure of the circuit $C_n$ and the labels of its gates can be computed in time $O(\poly n)$;
	\item
		has \emph{polynomial-time specifiable gates}: the coefficients in $\C$ of each gate in $C_n$, can be computed from the label of the gate in time $O(\poly n)$;
	\item
		decides $L$ \emph{exactly}: for any input $x$ of size $n$, $C_n \ket{x} = \ket{\psi(x)}\ket{0}$ for \NO\ instances, and $C_n \ket{x} = \ket{\psi(x)}\ket{1}$ for \YES\ instances, where $\ket{\psi(x)}$ is a pure state.
  \end{itemize}
\end{definition}
\noindent
Clearly $\EQP \subset \UnitaryP[\C]$ (by restricting to circuit-families with ``constant-time specifiable'' gates); one may also show that $\UnitaryP[\C] \subset \BQP$ (see \mbox{Ref.~\cite[\S3.3\,C]{Beaudrap-2014}}).
A simple modification of Ref.~\cite[Theorem~6.2]{ADH-1997} will show that the gates of the circuits of $\UnitaryP[\C]$ algorithms may be restricted to algebraic numbers.
As circuits with polytime-specifiable gates are considerably more flexible than those with constant-time specifiable gates --- for example, the former include quantum Fourier transforms of arbitrary order, while the latter does not --- we conjecture that $\EQP$ is strictly contained in $\UnitaryP[\C]$.

It is not obvious whether $\UnitaryP[\C] \subset \LPWPP$.
The bound $\EQP \subset \LPWPP$ relies on the gate-coefficients of the entire family $\{ C_n \}_{n\ge1}$ having a common denominator; in a circuit family $\{C_n\}_{n\ge1}$ constructed from gates from a potentially infinite basis, the set of coefficients (for all gates used by the family) may have no common denominator.
However, we may still bound $\UnitaryP[\C]$ by gap-definable classes.
The following follows by a straightforward modification of the proof of Ref.~\cite[Theorem~3.8]{FR-1999}, using the techniques of Ref.~\cite{ADH-1997} to represent algebraic amplitudes:
\begin{proposition}
	\label{prop:UnitaryPbound}
	$\UnitaryP[\C] \subset \LWPP$.  
\end{proposition}
\noindent
We consider this evidence that allowing ``potentially infinite'', but efficiently specifiable, gate-sets for quantum algorithms is not computationally extravagant.
(We present a foundational argument for the study of circuit families in the Appendix.)

\subsection{Quasi-quantum complexity}

Consider the following ``quasi-quantum'' models of computation, in which we allow non-unitary transformations of state vectors, following Aaronson~\cite{Aaronson-2005}.
\begin{definition}
	\label{def:quasiQuantumClasses}
  Define the following classes in analogy to $\EQP$, $\UnitaryP[\C]$, and $\ZQP$:
  \vspace*{-0.5ex}
  \begin{itemize}
  \item
		$\EQP_{\GL}$ is the set of problems which may be decided by polytime-uniform circuit families $\{C_n\}$, over a finite set of \emph{invertible} gates, such that $C_n \ket{x} = \ket{\psi(x)}\ket{0}$ if $x$ is a \NO\ instance and $C_n \ket{x} = \ket{\psi(x)}\ket{1}$ if $x$ is a \YES\ instance.
		\smallskip
  \item
		$\GLP[\C]$ is the set of problems which may be decided by polytime-uniform circuit families $\{C_n\}$, with polytime-specifiable \emph{invertible} gates, such that $C_n \ket{x} = \ket{\psi(x)}\ket{0}$ if $x$ is a \NO\ instance and $C_n \ket{x} = \ket{\psi(x)}\ket{1}$ if $x$ is a \YES\ instance.
		\smallskip
	\item
		$\ZQP_{\GL}$ is the set of problems which may be decided with zero error by polytime-uniform circuit families $\{C_n\}$, over a finite set of \emph{invertible} gates, and two sets of standard-basis projectors $\{\Pi_F,\Pi_S\}$ and $\{\Pi_0, \Pi_1\}$ on distinct qubits, and with \emph{constant success probability} in that for the renormalised state-vector $\ket{\Psi(x)} = C_n\!\!\:\ket{x} \big/ \sqrt{\smash{\bra{x} C_n\herm C_n \ket{x}}\phantom\vert\!}$, 
		\vspace*{-0.5ex}
		\begin{itemize}
		\item
			$\bra{\Psi(x)} \Pi_S \ket{\Psi(x)} \ge \tfrac{1}{2}$ for both \YES\ and \NO\ instances $x$;
		\item
			$\bra{\Psi(x)} \Pi_0 \Pi_S \ket{\Psi(x)} = 0$ for \YES\ instances;
		\item
			$\bra{\Psi(x)} \Pi_1 \Pi_S \ket{\Psi(x)} = 0$ for \NO\ instances.
		\end{itemize}
	\item
		$\Post\EQP$ is the set of problems which may be decided by polytime-uniform circuit families $\{C_n\}$, over a finite set of \emph{unitary} gates, and two sets of standard-basis projectors $\{\Pi_F,\Pi_S\}$ and $\{\Pi_0, \Pi_1\}$ on distinct qubits, and \emph{exactly with postselection} in the sense that for $\ket{\Psi(x)} = C_n\ket{x}$,
		\vspace*{-0.5ex}
		\begin{itemize}
		\item
			$\bra{\Psi(x)} \Pi_S \ket{\Psi(x)} > 0$ for both \YES\ and \NO\ instances $x$;
		\item
			$\bra{\Psi(x)} \Pi_0 \Pi_S \ket{\Psi(x)} = 0$ for \YES\ instances;
		\item
			$\bra{\Psi(x)} \Pi_1 \Pi_S \ket{\Psi(x)} = 0$ for \NO\ instances.
		\end{itemize}
  \end{itemize}
\end{definition}
The classes $\EQP_{\GL}$, $\GLP[\C]$, and $\ZQP_{\GL}$ are invertible-gate analogues of $\EQP$, $\UnitaryP[\C]$, and $\ZQP$ respectively, in which we renormalise the state before producing the output.
The class $\Post\EQP$ is a variant of $\EQP$, in which we post-select on the value of some other qubit being in the state $\ket{1}$ (ignoring the output in all other branches) prior to producing the output, and renormalise the state conditioned on the postselected outcome.
For those classes which decide languages exactly, the renormalisation has no effect on the decomposition of the result as a tensor product of the answer and the remaining qubits, and so is omitted from the definitions of those classes.

\vspace*{-2ex}
\paragraph{Remark.} The notations for the classes above are slightly non-uniform.
We use the notation $\GLP[\C]$ (as well as the notation $\UnitaryP[\C]$) as part of the framework for quasi-probabilistic (or ``modal'') computational models defined in Ref.~\cite{Beaudrap-2014}.
The ``$\mathsf{GL}$'' refers to the fact that the gates are elements of $\GL_k(\C)$ for various values of $k$.
We use the notations $\EQP_{\GL}$ and $\ZQP_{\GL}$ as a compromise between this convention and the notation $\EQP_{\!\mathrm{nu}}$ and $\ZQP_{\!\mathrm{nu}}$ which would extend the notation suggested by Aaronson~\cite{Aaronson-2005} for bounded-error ``quantum'' algorithms using invertible non-unitary gates.
One could analogously consider (non-unitary) quasi-quantum ``algorithms'', consisting of circuit families in which the polytime-specifiable gates are affine operators over $\C$, \label{discn:affine}\ie~which conserve the sum of the coefficients of the distributions on which they act (treating them as quasi-probability distributions).
One might then denote the corresponding complexity classes by \AffineP[\C], $\EQP_{\!\mathrm{Aff}}$, and $\ZQP_{\!\mathrm{Aff}}$\,.
(We do not study the latter classes here, but note their existence to justify the notation used for $\EQP_{\GL}$ and $\ZQP_{\GL}$\,.)

\medskip
As the proofs of Propositions~
\ref{prop:LPWPPunitaryBound} and~\ref{prop:UnitaryPbound} do not depend in any way on the transformations involved being unitary, the following proposition is implicit in Refs.~\cite{ADH-1997,FR-1999}:
\begin{proposition}
	\label{prop:exactInvertibleUpperBound}
  $\EQP_{\GL} \subset \LPWPP$ and $\GLP[\C] \subset \LWPP$.
\end{proposition}
\noindent
That is, allowing non-unitary operations does not allow such exact ``quasi-quantum'' algorithms to exceed the known upper bounds on exact quantum complexity.
Similarly, the results of Fenner~\etal~\cite{FGHP-1999} that $\NQP = \co\Ceq\P$ does not depend on the gates being unitary; we then have
\begin{proposition}
	\label{prop:ZQP-GL-upperBound}
	$\ZQP_{\GL} \subset \Delta\Ceq\P$.
\end{proposition}
\noindent 
Finally, given a $\Post\EQP$ algorithm for a problem, we may implement an $\NQP$ and a $\co\NQP$ algorithm for the same problem along the lines described in the proof of Proposition~\ref{prop:ZQPupperBound}, so that we have:
\begin{proposition}
	\label{prop:PostEQP-upperBound}
  $\Post\EQP \subset \Delta \Ceq\P$.
\end{proposition}
The main results of this article are to show that the containments of the above three propositions all hold with equality.

\section{Technical definitions in exact counting complexity}

We now define some concepts relating to algorithms for exact gap-definable classes.
This will simplify the analysis of simulations of these algorithms by quasi-quantum algorithms, of the sort described as part of Definition~\ref{def:quasiQuantumClasses}.

\subsection{Dual nondeterministic machines}

It will prove helpful to describe algorithms in terms of \emph{dual nondeterministic Turing machines}: a pair of normal-form NTMs $(\mathbf N_0, \mathbf N_1)$ for a language $L$, such that $\mathbf N_0$ represents a $\Ceq\P$ algorithm for $L$ (whose gap-function $g_0$ evaluates to zero for $x \in L$) and $\mathbf N_1$ represents a $\co\Ceq\P$ algorithm for $L$ (whose gap-function $g_1$ evaluates to zero for $x \notin L$).
We further require that for any input $x \in \{0,1\}^\ast$, pairs of dual NTMs make the same number of nondeterministic choices as one another.

\begin{lemma}
  Every $L \in \Ceq\P \cap \co\Ceq\P$ has a pair of dual NTMs.
\end{lemma}
\begin{proof}
	Let $\mathbf N_0$ be an NTM performing a $\Ceq\P$ algorithm for $L$ with $q_0(|x|)$ nondeterministic transitions in each branch, and $\mathbf N_1$ be an NTM performing a $\co\Ceq\P$ algorithm for $L$ with $q_1(|x|)$ nondeterministic transitions in each branch.
	We augment whichever machine $\mathbf N_j$ makes fewer transitions, as follows.
	Construct an NTM $\mathbf N'_j$ which has branching strings $b \in \{0,1\}^{\max\{q_0(|x|), q_1(|x|)\}}$, which first simulates $\mathbf N_j$, and then performs the non-deterministic transitions corresponding to the final $\bigl|q_1(|x|) - q_0(|x|)\bigr|$ bits of $b$, recording those bits as it does so.
	If those bits are of the form $1^\ast(0|1)$, then $\mathbf N'_j$ accepts if and only if $\mathbf N_j$ accepts; otherwise $\mathbf N'_j$ accepts if the parity of the substring is odd.
	The resulting gap-function $g'_j$ then satisfies $g'_j(x) = 2g_j(x)$ for all $x$, and the machines $N'_j$ and $N_{(1-j)}$ perform the same number of nondeterministic transitions.
\end{proof}
\noindent
It will also be useful to consider \emph{dual $\LWPP\!$ machines} (and \emph{dual $\LPWPP\!$ machines}): a pair $(\mathbf N_0,\mathbf N_1)$ of dual machines for some $L \in \LWPP$ (or $L \in \LPWPP$ respectively), for which there is a single poly-time computable, length-dependent, non-zero function $h$ (whose values are all powers of a fixed $M \ge 1$ for $L \in \LPWPP$) such that either $\tfrac{1}{2}g_0(x) = h(x)$ or $\tfrac{1}{2}g_1(x) = h(x)$.

\begin{lemma}
  Every $L \in \LWPP$ has a pair of dual \LWPP\! machines, and every $L \in \LPWPP$ has a pair of dual \LPWPP\! machines.
\end{lemma}
\begin{proof}
  Consider a poly-time NTM $\mathbf N$ in normal form with branching strings $b \in \{0,1\}^{q(|x|)}$, representing an \LWPP\ algorithm to decide a language $L$, in that the gap-function $g$ of $\mathbf N$ satisfies $g(x) = 0$ for $x \notin L$ and $g(x) = 2h(x)$ for $x \in L$, where $h$ is non-zero, efficiently computable, and length-dependent.  
  We may form a dual pair of NTMs $(\mathbf N_0, \mathbf N_1)$, as follows.
	We construct $\mathbf N_0$ to make a nondeterministic transition corresponding to a bit $\beta$.
	If $\beta = 0$, it nondeterministically guesses an integer $0 \le b' < 2^{q(|x|)}$, and rejects if $b' < 2^{q(|x|)-1} - h(x)$, accepting otherwise.
	Otherwise, for $\beta = 1$, $\mathbf N_0$ simulates $\mathbf N$, and accepts if and only if $\mathbf N$ rejects.
	If $g$ is the gap function for $\mathbf N$, the gap function for $\mathbf N_0$ is then $-g(x) + 2h(x)$, which is equal to $2h(x)$ if $x \notin L$ and $0$ otherwise.
	We similarly construct $\mathbf N_1$ to make a nondeterministic transition corresponding to a bit $\beta$; if $\beta = 0$ it accepts on exactly half of the branching strings, and if $\beta = 1$ it accepts if and only if a simulation of $\mathbf N$ accepts.
	Then $(\mathbf N_0,\mathbf N_1)$ are dual $\LWPP$ machines, with branching strings in $\{0,1\}^{q(|x|)+1}$, and gap functions governed by the same poly-time computable function $h$ which governs $\mathbf N$.
	In particular, if $L \in \LPWPP$, we may require that $h$ computes powers of some fixed integer $M \ge 1$, in which case $(\mathbf N_0, \mathbf N_1)$ are dual $\LPWPP$ machines.
\end{proof}

\subsection{Verifier functions and gap amplitudes}
\label{sec:verifierFunctions}%
\label{sec:gapAmplitudes}

It will prove more convenient in our analysis to refer to polytime-computable functions $f: \{0,1\}^\ast \x \{0,1\}^\ast \to \{0,1\}$, which compute the acceptance conditions of NTMs in normal form, rather referring to than the NTMs themselves.
In particular, if $\mathbf N$ is a poly-time NTM in normal form which halts in $q(n)$ steps for inputs of length $n \in \N$, we let 
$f(x,b) \in \{0,1\}$ represent the acceptance condition of $\mathbf N$ in the branch represented by $b$ on the input $x$.
(For a boolean string $b$ which is longer or shorter than $q(|x|)$, we suppose that $f(x,b) = 0$.)
\begin{subequations}
We then define:
\begin{align}
\!
	A(x,f,q) \,&=\, \#\Bigl\{ b \in \{0,1\}^{q(|x|)} \;\Big|\; f(x,b) = 1 \Bigr\},\!
\\[0.5ex]
\!
	R(x,f,q) \,&=\, \#\Bigl\{ b \in \{0,1\}^{q(|x|)} \;\Big|\; f(x,b) = 0 \Bigr\},\!
\intertext{%
and the (half-)gap function
}
	\Delta(x,f,q) \,&=\, \tfrac{1}{2}\Bigl[R(x,f,q) - A(x,f,p)\Bigr] \,=\, R(x,f,q) \:\!-\:\! 2^{q(|x|)\text{\;\!--\;\!}1}.
\end{align}
\end{subequations}
The functions $A(x,f,q)$ and $R(x,f,q)$ will largely occur in our analysis of amplitudes, preceded by a small scalar factor.
Furthermore, we are mainly interested in the case where $f$ computes the acceptance condition of one of a pair of dual NTMs.
We will typically consider ``dual pairs'' of verifier functions: boolean functions
\begin{align}
	V_n\sur{0}&: \{0,1\}^n \x \{0,1\}^{q(n)} \to \{0,1\}, &
	V_n\sur{1}&: \{0,1\}^n \x \{0,1\}^{q(n)} \to \{0,1\}
\end{align}
which compute the respective acceptance conditions of a pair of dual NTMs $(\mathbf N_0, \mathbf N_1)$ on inputs $x \in \{0,1\}^n$.
When the pair of dual NTMs are defined by context (with $q$ implicitly depending on them), we then write
\begin{align}
	\alpha_x\sur{c} &= \frac{\mbox{\small$A(x,V^{(c)},q)$}}{2^{q(|x|)}}\,,
	&
	\rho\sur{c}_x &= \frac{\mbox{\small$R(x,V^{(c)},q)$}}{2^{q(|x|)}}\,,
	&
	\delta\sur{c}_x &= \frac{\Delta(x,V\sur{c},q)}{2^{q(|x|)}}
\end{align}
for the sake of brevity.
(Note that $\rho\sur{c}_x + \alpha\sur{c}_x = 1$ by construction.)

We refer to $\delta\sur{0}_x, \delta\sur{1}_x$ as the \emph{gap amplitudes} of the NTMs $\mathbf N_0$ and $\mathbf N_1$; our simulation techniques largely concern evaluating such gap amplitudes.

\section{Quasi-quantum algorithms for $\Delta\Ceq\P$}

The definitions of the classes $\ZQP_{\GL}$ and $\Post\EQP$ in Definition~\ref{def:quasiQuantumClasses} are quite similar: the notional ``exactness'' of $\Post\EQP$ belies the fact that it involves a projective operation, which allows one to neglect unbounded error.
We show in this section, for zero-error quantum-like algorithms, how these two varieties of non-unitarity are equivalent.

\subsection{A unitary circuit to construct gap amplitudes}
\label{sec:genericAcceptanceGaps}

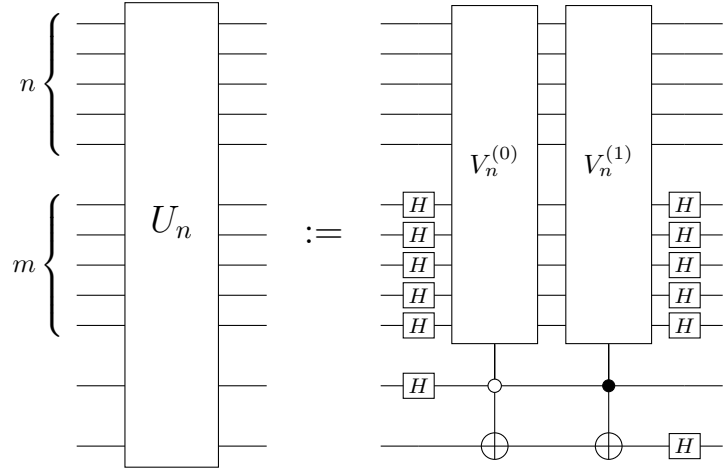
\begin{figure}[!t]
\begin{center}
\begin{tikzpicture}
		\def\J{1,2,3,4,5}
		\foreach \j in \J {%
			\ifnum\j=1
				\coordinate (x\j-0) at (0,0);
			\else
				\coordinate (x\j-0) at ($(\prev) + (0,-0.4)$);
			\fi
			\xdef\prev{x\j-0}
		}
				
		\foreach \j in \J {%
			\ifnum\j=1
				\coordinate (b\j-0) at ($(\prev) + (0,-0.8)$);
			\else
				\coordinate (b\j-0) at ($(\prev) + (0,-0.4)$);
			\fi
			\xdef\prev{b\j-0}
		}
					
		\coordinate (c-0) at ($(\prev) + (0,-0.8)$);
		\xdef\prev{c-0}

		\coordinate (a-0) at ($(\prev) + (0,-0.8)$);
		\xdef\prev{a-0}
		

		\foreach \t/\Dt in {1/-4,2/-2.75,3/-1.5} {
			\foreach \l in {x,b} {%
				\foreach \j in \J {%
					\coordinate (U\l\j-\t) at ($(\l\j-0) + (\Dt,0)$);
					\ifnum\t>1
						\draw (U\l\j-\u) -- (U\l\j-\t);
					\fi
				}
			}
			\foreach \l in {c,a} {%
				\coordinate (U\l-\t) at ($(\l-0) + (\Dt,0)$);
					\ifnum\t>1
						\draw (U\l-\u) -- (U\l-\t);
					\fi
			}
			\xdef\u{\t}
		}
		\node (Ux1-2) at (Ux1-2) {$X$};
		\node (Ua-2) at (Ua-2) {$X$};
		\node (U) [draw,fill=white,inner sep=1pt,minimum width=3.5em,fit=(Ux1-2)(Ua-2)]
			{{\Large$\!\!\!U_n\!\!\!$}};
		\node (Ux) at ($(Ux1-1)!0.5!(Ux5-1)$) [anchor=east] {$n \left\{\mathclap{\begin{matrix} \\[10ex] \end{matrix}}\right.$};
		\node (Ub) at ($(Ub1-1)!0.5!(Ub5-1)$) [anchor=east] {$m \left\{\mathclap{\begin{matrix} \\[10ex] \end{matrix}}\right.$};

		\node at ($(x1-0)!0.5!(Ua-3)$) {\Large$:=$};
		
		\foreach \u/\t/\dt in {0/1/0.5,1/2/1,2/3/1.5,3/4/1,4/5/0.5}{
			\foreach \l in {x,b} {%
				\foreach \j in \J {%
					\coordinate (\l\j-\t) at ($(\l\j-\u) + (\dt,0)$);
					\draw [black] (\l\j-\u) -- (\l\j-\t);
				}
			}
			\foreach \l in {c,a} {%
				\coordinate (\l-\t) at ($(\l-\u) + (\dt,0)$);
				\draw [black] (\l-\u) -- (\l-\t);
			}
		}
		
		\foreach \j in \J {	\node [draw,fill=white,inner sep=2pt] at (b\j-1) {\footnotesize$H$}; }
		\node [draw,fill=white,inner sep=2pt] at (c-1) {\footnotesize$H$}; 
		
		\node (x1-2) at (x1-2) {};
		\node (b5-2) at (b5-2) {};
		\node (R) at ($(x1-2)!0.5!(b5-2)$) {$ V^{(c)}_n $};
		\node (R) [draw,fill=white,fit=(x1-2)(b5-2)(R)] {$ V^{(0)}_n $};
		\draw [black] (a-2) circle (5pt);
		\draw (c-2) -- (R) -- ($(a-2) + (0,-5pt)$);
		\fill [white] (c-2) circle (2.5pt);
		\draw [black] (c-2) circle (2.5pt);
	
		\node (x1-3) at (x1-3) {};
		\node (b5-3) at (b5-3) {};
		\node (R) at ($(x1-3)!0.5!(b5-3)$) {$ V^{(c)}_n $};
		\node (R) [draw,fill=white,fit=(x1-3)(b5-3)(R)] {$ V^{(1)}_n $};
		\draw [black] (a-3) circle (5pt);
		\filldraw [black] (c-3) circle (2.25pt);
		\draw (c-3) -- (R) -- ($(a-3) + (0,-5pt)$);

		\foreach \j in \J {		\node [draw,fill=white,inner sep=2pt] at (b\j-4) {\footnotesize$H$}; }
		\node [draw,fill=white,inner sep=2pt] at (a-4) {\footnotesize$H$};



\end{tikzpicture}
\end{center}
\caption{%
	Diagram for a unitary circuit $U_n$, to produce gap amplitudes corresponding to two a pair of dual nondeterministic Turing machines $(\mathbf N_0, \mathbf N_1)$.
	The functions $V_n^{(0)}$ and $V_n^{(1)}$ compute whether $\mathbf N_0$ and $\mathbf N_1$ (respectively) accept on an input $x \in \{0,1\}^n$, in a given computational branch labelled by $b \in \{0,1\}^m$ for $m = q(n) \in O(\poly n)$.
	The function $V^{(0)}_n$ is (coherently) conditioned on the second-to-last bit being in the state $|0\rangle$, while $V^{(1)}_n$ is conditioned on it being in the state $|1\rangle$; in each case the result is XORed onto the final bit.
	The final Hadamard operations serve to accumulate amplitudes corresponding to the accepting or rejecting branches of $\mathbf N_0$ and $\mathbf N_1$, and compute gap amplitudes for each machine $\mathbf N_j$.
	The component in which the ``branching register'' is $|0^m\rangle$ and the final bit is $|1\rangle$ then has an amplitude proportional to a gap function.
}
\label{fig:unitaryCircuit}
\end{figure}

Consider a dual pair of verifier functions $V_n\sur{0}$ and $V_n\sur{1}$ as in Section~\ref{sec:verifierFunctions}, corresponding to dual NTMs $(\mathbf N_0, \mathbf N_1)$ acting on inputs of length $n$, and let $m = q(n)$ be the size of the branching strings accepted by these verifiers.
Figure~\ref{fig:unitaryCircuit} presents a unitary circuit $U_n$ to produce the gap amplitudes of $(\mathbf N_0, \mathbf N_1)$ as amplitudes of a pure state.
Consider the evolution of the qubits in the circuit when presented with an input of $\ket{x}\ket{0^{m+2}}$.
After the first round of Hadamard operations, followed by coherently evaluating both verifier circuits, the state of the system may be expressed by
\begin{equation}
	\ket{\Psi_1(x)} = \frac{1}{\sqrt{2^{m+1}}}\mathop{\sum \sum}_{\substack{\mathclap{b \in \{0,1\}^m} \\ c\in\{0,1\}}} \ket{x}\ket{b}\ket{c}\ket{\smash{V^{(c)}(x,b)}}.
\end{equation}
\begin{subequations}
After the second round of Hadamards on the branching register (not counting the Hadamard on the final bit), we obtain the state
\begin{equation}
\begin{aligned}[b]
	\label{eqn:penultimate-Un-a}
		\ket{\Psi_2(x)}
	={}& \frac{1}{2^{m}\sqrt{2}}\mathop{\sum \sum \sum}_{\substack{\mathclap{\,\,b,z \in \{0,1\}^m\!\!} \\ c\in\{0,1\}}} \,(-1)^{z \cdot b} \ket{x}\ket{z}\ket{c}\ket{\smash{V^{(c)}(x,b)}}
	\\[0.5ex]{}={}&
		\ket{\Psi'_2(x)} + \frac{1}{2^{m}\sqrt{2}}\mathop{\sum \sum}_{\substack{\mathclap{\,\,b \in \{0,1\}^m\!\!} \\ c\in\{0,1\}}} \; \ket{x}\ket{0^m}\ket{c}\ket{\smash{V^{(c)}(x,b)}},
\end{aligned}
\end{equation}
where $\ket{\Psi'_2(x)}$ is simply the contributions to $\ket{\Psi_2(x)}$ for $z \ne 0^n$:
\begin{equation}
\begin{aligned}[b]
	\label{eqn:simulationFailComponent-1}
		\ket{\Psi'_2(x)}
	={}& \frac{1}{2^{m}\sqrt{2}}\,\mathop{\sum \sum \sum}_{\substack{\mathclap{\,\,\,b\!\!\;,z \in \{0,1\}^{\!\!\:m}\!\!,\;\!z\ne0^m\!\!} \\ c\in\{0,1\}}} \; (-1)^{z \cdot b} \ket{x}\ket{z}\ket{c}\ket{\smash{V^{(c)}(x,b)}}.
\end{aligned}
\end{equation}
\end{subequations}
In particular, we have $(\idop\sox{n} \ox \bra{0^m} \ox \idop\sox{2}) \ket{\Psi'_2(x)} = \vec 0$.
Evaluating the sum over $b \in \{0,1\}^m$ in the final term of Eqn.~\eqref{eqn:penultimate-Un-a}, we may then express
\begin{equation}
	\begin{aligned}[b]
	\label{eqn:penultimate-Un-b}
	\ket{\Psi_2(x)} :={}&
		\ket{\Psi'_2(x)}
		+ \tfrac{1}{\sqrt{2}}\ket{x}\ket{0^m} \sum_{\mathclap{c\in\{0,1\}}}\;\Bigl(
			\frac{\mbox{\small$R(x,V^{(c)}\!,m)$}}{2^{m}}\ket{c\;\!0} + \frac{\mbox{\small$A(x,V^{(c)}\!,m)$}}{2^{m}}\ket{c1} \Bigr)
	\mspace{-36mu}
	\\[1ex]{}={}&
				\ket{\Psi'_2(x)}
				+ \tfrac{1}{\sqrt 2} \ket{x} \ket{0^m} \sum_{\mathclap{c\in\{0,1\}}}\;\Bigl(
					\rho\sur{c}_x\!\!\;\ket{c\;\!0} + \alpha\sur{c}_x\!\!\;\ket{c1} \Bigr) .
	\end{aligned}
\end{equation}
To produce the gap amplitude in the amplitudes of the components, we perform the Hadamard operation on the final bit, yielding
\begin{equation}
	\begin{aligned}[b]
	\mspace{-12mu}
	\ket{\psi_x} :=
	\ket{\Psi_3(x)} ={}&
				\ket{x} \ket{0^m} \, \sum_{\mathclap{c\in\{0,1\}}} \:\: \biggl(
				\biggl[\frac{\rho\sur{c}_x\!+\alpha\sur{c}_x}{2}\biggr]\!\!\;\ket{c\;\!0} + \biggl[\frac{\rho\sur{c}_x\!-\alpha\sur{c}_x}{2}\biggr]\!\!\;\ket{c1} \biggr)
		\\[-1.5ex]&
				\phantom{\ket{x} \ket{0^m} \, \sum_{\mathclap{c\in\{0,1\}}} \:\: \biggl(
				\biggl[\frac{\rho\sur{c}_x\!+\alpha\sur{c}_x}{2}\biggr]\!\!\;\ket{c0} {}}
					{} + \Bigl(\idop\sox{n+m+1} \ox H\Bigr) \ket{\Psi'_2(x)}
					\mspace{-30mu}
		\\
		{}={}&
				\ket{x} \ket{0^m} \Bigl(\tfrac{1}{2} \ket{00} + \tfrac{1}{2} \ket{10} + \delta\sur{0}_x \ket{01} + \delta\sur{1}_x \ket{11} \Bigr)
		+ \ket{\psi'_x},			
	\end{aligned}
\end{equation}
where $\ket{\psi'_x} = \bigl(\idop\sox{n+m+1} \ox H\bigr) \ket{\Psi'_2(x)}$, and satisfies $(\idop\sox{n} \ox \bra{0^m} \ox \idop\sox{2}) \ket{\psi'_x} = \vec 0$.
Thus, in the component of $\ket{\psi_x} = U_n \ket{x}\ket{0^{m+2}}$ where the branching register is in the state $\ket{0^m}$ and the final bit is in the state $\ket{1}$, the amplitudes are given by the gap amplitudes of the machines $\mathbf N_0$ and $\mathbf N_1$.
Furthermore, by the fact that $(\mathbf N_0,\mathbf N_1)$ are dual, exactly one of $\delta\sur{0}_x, \delta\sur{1}_x$ is non-zero.
Specifically, $\delta\sur{0}_x = 0$ if $x \in L$, and $\delta\sur{1}_x = 0$ if $x \notin L$.
If we use the notation $L(x) = 0$ to indicate $x \notin L$ and $L(x) = 1$ to indicate $x \in L$, we then have $\delta\sur{L(x)}_x \ne 0$, and we may express
\begin{equation}
	\begin{aligned}[b]
	\label{eqn:final-Un}
	\ket{\psi_x} :=
				\ket{x} \ket{0^m} \Bigl(\tfrac{1}{2} \ket{00} + \tfrac{1}{2} \ket{10} + \delta\sur{L(x)}_x \ket{L(x)}\ket{1} \Bigr)
		+ \ket{\psi'_x}.
	\end{aligned}
\end{equation}
Note that $\bracket{\psi_x}{\psi_x} = 1$, as $U_n$ is unitary, and that measuring $\ket{\psi_x}$ in the standard basis yields ${\ket{x}\ket{0^m}\ket{00}}$ or ${\ket{x}\ket{0^m}\ket{10}}$ with probability $\tfrac{1}{4}$ each.
It then follows that $\bracket{\psi'_x}{\psi'_x} < \tfrac{1}{2}$.

\subsection{Quasi-quantum algorithms by amplitude forcing}
\label{sec:ZQPGL-and-PostEQP}

The above construction produces a state in which the gap amplitudes indicate (albeit possibly with low probability) whether $x \in L$ for an input $x$.
This will allow us to easily prove that $L \in \ZQP_{\GL}$ and $L \in \Post\EQP$, and that therefore $\Delta\Ceq\P = \ZQP_{\GL} = \Post\EQP$.

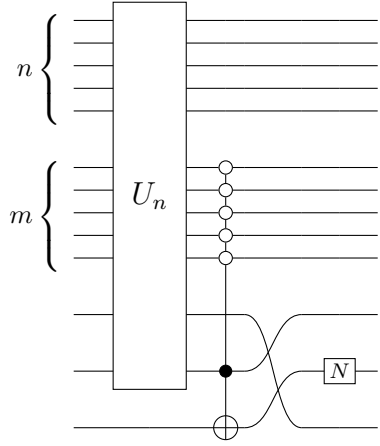
\begin{figure}
\begin{center}
\begin{tikzpicture}
		\def\J{1,2,3,4,5}
		\foreach \j in \J {%
			\ifnum\j=1
				\coordinate (x\j-0) at (0,0);
			\else
				\coordinate (x\j-0) at ($(\prev) + (0,-0.3)$);
			\fi
			\xdef\prev{x\j-0}
		}
				
		\foreach \j in \J {%
			\ifnum\j=1
				\coordinate (b\j-0) at ($(\prev) + (0,-0.75)$);
			\else
				\coordinate (b\j-0) at ($(\prev) + (0,-0.3)$);
			\fi
			\xdef\prev{b\j-0}
		}
					
		\coordinate (c-0) at ($(\prev) + (0,-0.75)$);
		\xdef\prev{c-0}

		\coordinate (a-0) at ($(\prev) + (0,-0.75)$);
		\xdef\prev{a-0}
		
		\coordinate (s-0) at ($(\prev) + (0,-0.75)$);
		\xdef\prev{s-0}
		
		\xdef\u{0}
		\foreach \t/\dt in {1/1,2/1,3/0.25,4/0.75,5/0.5,6/0.5}{
			\foreach \l in {x,b} {%
				\foreach \j in \J {%
					\coordinate (\l\j-\t) at ($(\l\j-\u) + (\dt,0)$);
					\draw [black] (\l\j-\u) -- (\l\j-\t);
				}
			}
			\foreach \l in {c,a,s} {%
				\coordinate (\l-\t) at ($(\l-\u) + (\dt,0)$);
			}
			\ifnum\t=4
				\coordinate (a-\u-5) at ($(a-\u)!0.5!(a-\t)$);
				\coordinate (s-\u-5) at ($(s-\u)!0.5!(s-\t)$);
				\coordinate (c-\u-5) at ($(c-\u)!0.5!(c-\t)$);
				\draw (a-\u) .. controls (a-\u-5) and (c-\u-5) .. (c-\t);
				\draw (c-\u) .. controls (c-\u-5) and (s-\u-5) .. (s-\t);
				\draw (s-\u) .. controls (s-\u-5) and (a-\u-5) .. (a-\t);
				\coordinate (temp) at (s-\t);
				\coordinate (s-\t) at (a-\t);
				\coordinate (a-\t) at (c-\t);
				\coordinate (c-\t) at (temp);
			\else
				\foreach \l in {c,a,s} {%
					\draw [black] (\l-\u) -- (\l-\t);
				}
			\fi
			\xdef\u{\t}
		}
			
		\node (x) at ($(x1-0)!0.5!(x4-0) + (0,-0.2)$) [anchor=east] {$n \left\{\mathclap{\begin{matrix} \\[7ex] \end{matrix}}\right.$};
		\node (b) at ($(b1-0)!0.5!(b4-0) + (0,-0.2)$) [anchor=east] {$m \left\{\mathclap{\begin{matrix} \\[7ex] \end{matrix}}\right.$};

		\node (x1-1) at (x1-1) {};
		\node (a-1) at (a-1) {};
		\node (U) at ($(x1-1)!0.5!(a-1)$) {\large$ U_n $};
		\node [draw,fill=white,fit=(x1-1)(a-1)(U)] {};
		\node (U) at (U) {\large$ U_n $};


		\filldraw [black] (a-2) circle (2.25pt);
		\draw [black] (s-2) circle (4.5pt);
		\draw ($(s-2) + (0,-4.5pt)$) -- (b1-2);
		\foreach \j in \J {	\filldraw [white] (b\j-2) circle (2.5pt); \draw [black] (b\j-2) circle (2.5pt); }

		\node [draw,fill=white,inner sep=2pt] at (s-5) {\footnotesize$N$};
		
\end{tikzpicture}
\end{center}
\caption{%
	Diagram for a circuit for an arbitrary language $L \in \Delta \Ceq\P$, consisting of unitary gates together with one non-unitary operation $N = \diag(p,1)$, for $0 \le p < 1$.
	This circuit implements either a $\ZQP_{\GL}$ algorithm or a $\Post\EQP$ algorithm, depending on the choice of $p$.
	The unitary circuit $U_n$ is as depicted in Figure~\ref{fig:unitaryCircuit}; the multiply-controlled NOT gate is conditioned on the its final control being in the state $|1\rangle$ and its other controls being in the state $|0^m\rangle$.
	\textbf{(a)}~For $p = 2^{-m}$, we have $N = B^{m}$ for $B = \diag(\tfrac{1}{2},1)$, so that the above circuit may be simulated using only classical operations, Hadamards, and $B$ gates.
	In this case, the second-to-last qubit produces the state $|1\rangle$ with probability greater than $\tfrac{1}{2}$, in which case the final bit is $|L(x)\rangle$ with certainty.
	\textbf{(b)}~For $p = 0$, we have $N = |1\rangle\langle1|$ and the above circuit describes a unitary circuit with a single post-selection on the second-to-last qubit.
	This post-selection gives rise to a well-defined state in which the final qubit is $|L(x)\rangle$ with certainty.
}
\label{fig:zeroErrorCircuit}
\end{figure}

The second-to-last bit of $\ket{\psi_x}$ stores $L(x)$ in the component where the branching register takes the state $\ket{0^m}$ and the final bit is in the state $\ket{1}$.
By using a multiply-controlled NOT conditioned on these states, we may indicate in a single bit whether or not $L(x)$ has been successfully computed.
Consider the circuit in Figure~\ref{fig:zeroErrorCircuit}, which computes such a bit, performs a non-unitary single-bit operation
$N = \diag(p,1)$
for some $0 \le p < 1$, and performs a permutation so that the final bit presents (with some probability) the outcome $L(x)$.
Given an input of $\ket{x}\ket{0^{m+3}}$ the output of this circuit is
\begin{equation}
	\label{eqn:ZQPandPostEqpFinalState}
  \ket{\Psi_{\!\!\:L,x}} = 
				\ket{x} \ket{0^m} \Bigl(\tfrac{p}{2} \ket{000} + \tfrac{p}{2} \ket{001} + \delta\sur{L(x)}_x \ket{11} \ket{L(x)} \Bigr) + p \ket{\psi''_x},			
\end{equation}
where $\ket{\psi''_x}$ is the result of cyclically permuting the final three qubits of $\ket{\psi'_x}\ket{0}$.
(In particular, the second-to-last qubit is in the state $\ket{0}$.)
We may then describe $\ZQP_{\GL}$ and $\Post\EQP$ algorithms for $L$, as follows.

\subsubsection{Zero-error algorithms by invertible gap amplification}

To consider a zero-error algorithm, consider the projectors ${\Pi_F = \ket{0}\bra{0}}$ and ${\Pi_S = \ket{1}\bra{1}}$ performed on the second-to-last qubit, and consider the effect of a projective $\{\Pi_F, \Pi_S\}$ measurement.
We interpret an outcome of $\Pi_F$ as a failure of the algorithm to produce an answer; on an outcome of $\Pi_S$, we produce the value of the final bit as output of the algorithm.
By construction, when we obtain the outcome $\Pi_S$, the value of the final bit will be $\ket{L(x)}$.
This is then a zero-error algorithm for any value of $p$.
To ensure that this zero-error algorithm has a bounded probability of failure, it suffices to bound the probability of obtaining the outcome $\Pi_F$ from above.

To determine the probability of any particular outcome from $\ket{\Psi_{\!\!\:L,x}}$, we renormalise the vector and then apply the Born rule.
Thus the probability of any measurement outcome is determined by the ratio of the modulus-squared of its amplitude, relative to that of all other amplitudes.
We have $\bracket{\psi''_x}{\psi''_x} < \tfrac{1}{2}$ from the remarks at the end of Section~\ref{sec:genericAcceptanceGaps}; thus
\begin{equation}
  \Bigl\| \Pi_F \ket{\Psi_{\!\!\:L,x}} \Bigr\|^2 \;<\; \frac{p^2}{4} + \frac{p^2}{4} + \frac{p^2}{2} \;=\; p^2,
\end{equation}
whereas $\bigl\| \Pi_S \ket{\Psi_{\!\!\:L,x}} \bigr\|^2 = \delta\sur{L(x)}_x\big.^{\;2}$.
Thus we have
\begin{equation}
  \Bigl\| \Pi_F \ket{\Psi_{\!\!\:L,x}} \Bigr\|^2 \bigg/ \Bigl\| \Pi_S \ket{\Psi_{\!\!\:L,x}} \Bigr\|^2 <
	\Bigl(p \big/ \delta\sur{L(x)}_x \Bigr)^2\,.
\end{equation}
If we bound the right-hand side from above by $1$, the probability of failure will then be less than $\tfrac{1}{2}$.
Given that $\delta\sur{L(x)}_x \ge 2^{-m}$, we may let 
$
	p \le 2^{-m}
$, 
which suffices to ensure $p \le \delta_x\sur{L(x)}$.
If we let $N = B^{m}$, where 
\begin{equation}
	\label{eqn:probabilityBoostOptor}
	B = \begin{bmatrix} \tfrac{1}{2} \!&\! 0\; \\[1ex] 0 \!&\! 1\; \end{bmatrix},
\end{equation}
then $p = 2^{-m}$, so that $\Pi_F$ occurs with probability less than $\tfrac{1}{2}$.
Thus we have a zero-error polytime algorithm (using $2m \in O(\poly n)$ non-unitary invertible gates) with a bounded probability of failure, so that $L \in \ZQP_{\GL}$.
By Proposition~\ref{prop:ZQP-GL-upperBound}, we then have:

\begin{theorem}\label{thm:ZQPGL}
  $\ZQP_{\GL} = \Delta\Ceq\P$.
\end{theorem}

\subsubsection{Forcing exactness by post-selection} 

The above zero-error algorithm suppresses the probability of failure by a scalar factor $p \ll 1$.
Taking this idea to its logical limit (\ie~setting $N = B^t$ and letting $t \to \infty$) is equivalent to postselecting the value $\ket{1}$ on the second-to-last qubit of $\ket{\Psi_{\!\!\:L,x}}$.
In other words, taking $N = \ket{1}\bra{1}$ yields an algorithm with a well-defined final state after renormalisation, and in which the final bit is certain to be in the state $\ket{L(x)}$.
This yields an exact polynomial-time (but postselective) algorithm for $L$, so that $L \in \Post\EQP$.
Together with Proposition~\ref{prop:PostEQP-upperBound}, we then have:
\begin{theorem}\label{thm:PostEQP}
	$\Post\EQP = \Delta \Ceq\P$. 
\end{theorem}

\section{Quasi-quantum algorithms for $\LWPP$ and $\LPWPP$}

Postselected quantum computations can achieve exactness, in effect, by the ignoring the probability of failure of a zero-error algorithm.
Achieving exact quantum algorithms with invertible operations is subtler, as we show in this section.
The set of problems which exact nonunitary algorithms can solve is apparently smaller, even if one considers algorithms with a possibly infinite (but polynomial-time specifiable) gate-set.

\subsection{An invertible circuit to construct gap amplitudes}

We first define a circuit $W_n$ as in Figure~\ref{fig:skew}, which builds upon the unitary circuit $U_n$ described in Figure~\ref{fig:unitaryCircuit} and uses a single non-unitary gate
\begin{equation}
	\label{eqn:elemRowOpn}
  S = \mbox{\small$\begin{bmatrix} 1 & 1 \\ 0 & 1 \end{bmatrix}$}.
\end{equation}%
This circuit constructs the gap amplitude $\delta_x\sur{L(x)}$ for a component of its output in much the same way that $U_n$ does, and also uses $U_n\herm$ to uncompute the ``garbage'' components contained in the state $\ket{\Psi_{\!\!\:L,x}}$ described in Eqn.~\eqref{eqn:ZQPandPostEqpFinalState}.
This will simplify the analysis to follow for exact invertible algorithms.
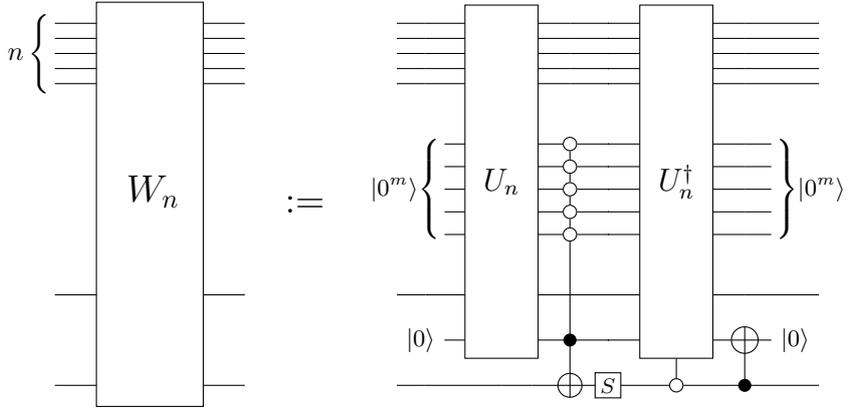
\begin{figure}
\begin{center}
\hspace*{-2.2em}
\begin{tikzpicture}
		\def\J{1,2,3,4,5}
		\foreach \j in \J {%
			\ifnum\j=1
				\coordinate (x\j-0) at (0,0);
			\else
				\coordinate (x\j-0) at ($(\prev) + (0,-0.2)$);
			\fi
			\xdef\prev{x\j-0}
		}

		\foreach \j in \J {%
			\ifnum\j=1
				\coordinate (b\j-0) at ($(\prev) + (0,-0.8)$);
			\else
				\coordinate (b\j-0) at ($(\prev) + (0,-0.3)$);
			\fi
			\xdef\prev{b\j-0}
		}
					
		\coordinate (c-0) at ($(\prev) + (0,-0.8)$);
		\xdef\prev{c-0}

		\coordinate (a-0) at ($(\prev) + (0,-0.6)$);
		\xdef\prev{a-0}
	
		\coordinate (s-0) at ($(\prev) + (0,-0.6)$);
		\xdef\prev{s-0}


		\xdef\u{0}
		\foreach \t/\dt in {1/1,2/0.4,3/0.5,4/0.5,5/0,6/0,7/0.9,8/0,9/0.9,10/0.5} {
			\foreach \j in \J {%
				\coordinate (x\j-\t) at ($(x\j-\u) + (\dt,0)$);
				\draw [black] (x\j-\u) -- (x\j-\t);
				\coordinate (b\j-\t) at ($(b\j-\u) + (\dt,0)$);
				\ifnum\u<1
					\coordinate (b\j-0) at ($(b\j-0) + (0.25,0)$);
				\else\ifnum\t>9
					\coordinate (b\j-10) at ($(b\j-10) + (-0.15,0)$);
				\fi\fi
				\draw [black] (b\j-\u) -- (b\j-\t);
			}
			\foreach \l in {c,a,s} {%
				\coordinate (\l-\t) at ($(\l-\u) + (\dt,0)$);
				\if a\l
					\ifnum\u<1
						\coordinate (a-0) at ($(a-0) + (0.25,0)$);
					\else\ifnum\t>9
						\coordinate (a-10) at ($(a-10) + (-0.15,0)$);
				\fi\fi\fi
				\draw [black] (\l-\u) -- (\l-\t);
			}
			\xdef\u{\t}
		}
			
		\node (x1-1) at (x1-1) {};
		\node (a-1) at (a-1) {};
		\node (U) at ($(x1-1)!0.5!(a-1)$) {\large$ U_n $};
		\node [draw,fill=white,fit=(x1-1)(a-1)(U)] {};
		\node (U) at (U) {\large$ U_n $};

 		
		\filldraw [black] (a-3) circle (2.25pt);
		\draw [black] (s-3) circle (4.5pt);
		\draw ($(s-3) + (0,-4.5pt)$) -- (b1-3);
		\foreach \j in \J {	\filldraw [white] (b\j-3) circle (2.5pt); \draw [black] (b\j-3) circle (2.5pt); }

		\node [draw,fill=white,inner sep=2pt] at (s-4) {\footnotesize$S$};



		\node (x1-7) at (x1-7) {};
		\node (a-7) at (a-7) {};
		\node (U) at ($(x1-7)!0.5!(a-7)$) {\large$ U_n $};
		\node (U) [draw,fill=white,fit=(x1-7)(a-7)(U)] {};
		\node at ($(x1-7)!0.5!(a-7)$) {\large$ U_n\herm $};
		\draw (s-7) -- (U);
		\filldraw [white] (s-7) circle (2.5pt);
		\draw [black] (s-7) circle (2.5pt);
				

		\filldraw [black] (s-9) circle (2.25pt);
		\foreach \l in {a}{
			\filldraw [white] (\l-9) circle (5pt);
			\draw [black] (\l-9) circle (5pt);
			\draw ($(\l-9) + (-5pt,0)$) -- ++(10pt,0);
		}
		\draw (s-9) -- ($(a-9) + (0,5pt)$); 

		\node (ai) at (a-0) [anchor=east] {\small$\ket{0}$};		
		\node (bi) at ($(b1-0)!0.5!(b5-0)$) [anchor=east] {\small$\;\;\ket{0^m} \!\left\{\mathclap{\begin{matrix} \\[7ex] \end{matrix}}\right.\!\!$};		

		\node (af) at (a-\u) [anchor=west] {\small$\ket{0}$};		
		\node (bf) at ($(b1-\u)!0.5!(b5-\u)$) [anchor=west] {\small$\!\!\left. \mathclap{\begin{matrix} \\[7ex] \end{matrix}}\right\}\! \ket{0^m}\;\;$};

		\foreach \j in \J {%
			\draw (x\j-0) -- (x\j-0 -| ai.west);
			\draw (x\j-\u) -- (x\j-\u -| af.east);
			\coordinate (x\j-0) at (x\j-0 -| ai.west);
			\coordinate (x\j-\u) at (x\j-\u -| af.east);
		}
		\foreach \l in {c,s} {%
			\draw (\l-0) -- (\l-0 -| ai.west);
			\draw (\l-\u) -- (\l-\u -| af.east);
			\coordinate (\l-0) at (\l-0 -| ai.west);
			\coordinate (\l-\u) at (\l-\u -| af.east);
		}

		\foreach \t/\Dt in {1/-4.5,2/-3.25,3/-2} {
			\foreach \l in {x} {%
				\foreach \j in \J {%
					\coordinate (U\l\j-\t) at ($(\l\j-0) + (\Dt,0)$);
					\ifnum\t>1
						\draw (U\l\j-\u) -- (U\l\j-\t);
					\fi
				}
			}
			\foreach \l in {c,s} {%
				\coordinate (U\l-\t) at ($(\l-0) + (\Dt,0)$);
					\ifnum\t>1
						\draw (U\l-\u) -- (U\l-\t);
					\fi
			}
			\xdef\u{\t}
		}
		\node (Ux1-2) at (Ux1-2) {$X$};
		\node (Us-2) at (Us-2) {$X$};
		\node (U) [draw,fill=white,inner sep=1pt,minimum width=4em,fit=(Ux1-2)(Us-2)]
			{{\Large$\!\!\!\!W_n\!\!\!\!\!$}};
		\node (Ux) at ($(Ux1-1)!0.5!(Ux5-1)$) [anchor=east] {$n \left\{\mathclap{\begin{matrix} \\[4ex] \end{matrix}}\right.$\!\!\!\:};

		\node at ($(x1-0)!0.5!(Us-3)$) {\Large\!\!\!$:=$\,\,};
\end{tikzpicture}
\hspace*{-2.2em}
\end{center}
\caption{%
	Diagram for a circuit $W_n$ consisting of unitary gates, together with one non-unitary operation $S$, to evaluate gap functions corresponding to two distinct poly-time nondeterministic Turing machines $\mathbf N_0$ and $\mathbf N_1$.
	The circuit $U_n$, and its relationship to the two nondeterministic machines $\mathbf N_0$ and $\mathbf N_1$, are presented in Figure~\ref{fig:unitaryCircuit}.
	The $S$ gate is a non-unitary transformation given by $S = \bigl[\protect\begin{smallmatrix} 1 & 1 \\ 0 & 1 \protect\end{smallmatrix}\bigr]$.
	Operations with controls marked by white circles are (coherently) conditioned on the qubit being in the state $|0\rangle$; operations with controls marked by black circles are conditioned on the qubit being in the state $|1\rangle$
	The qubits with explicitly indicated input and output states are ancillas.
}
\label{fig:skew}
\end{figure}%

As before, consider a dual pair of verifier functions $V_n\sur{0}$ and $V_n\sur{1}$, corresponding to dual NTMs $(\mathbf N_0, \mathbf N_1)$ acting on inputs of length $n$, and let $m = q(n)$ be the size of the branching strings accepted by these verifiers.
Let $U_n$ be as described in Figure~\ref{fig:unitaryCircuit}, defined in terms of functions $V_n\sur{0}$ and $V_n\sur{1}$.
On an input of $\ket{x}\ket{0}\ket{0}$, the circuit first introduces several ancilla bits to produce $\ket{x}\ket{0^{m+3}}$.
Acting on this with the operator $U_n \ox \idop$ then yields a state
\begin{equation}
	\begin{aligned}[b]
		\ket{\Phi_1(x)} \,&=\,	\ket{\psi_x} \ket{0}
		\\&
			\,=\,
			\ket{x} \ket{0^m} \Bigl( \tfrac{1}{2}\ket{00} + \tfrac{1}{2}\ket{10} + \delta\sur{L(x)}_x\ket{L(x)}\ket{1} \Bigr)\ket{0}
			+ \ket{\psi'_x} \ket{0},			
	\end{aligned}
\end{equation}
following Eqn.~\eqref{eqn:final-Un}.
In particular, we have $(\idop\sox{n} \ox \bra{0^m} \ox \idop\sox{2}) \ket{\psi'_x} = \vec 0$.
Then after the multiply controlled-NOT gate, we obtain the state
\begin{equation}
	\!\!\!\!\!\!
	\begin{aligned}[b]
		\ket{\Phi_2(x)} ={}&	
			\ket{x} \!\ket{0^m} \!\Bigl( \!\!\;\tfrac{1}{2}\!\!\:\ket{00} \!\!\;{+} \tfrac{1}{2}\!\!\:\ket{10} \!\!\;\Bigr) \!\!\!\;\ket{0} + \ket{\psi'_x} \! \ket{0}
		+\delta\sur{L(x)}_x \ket{x} \! \ket{0^m} \! \ket{L(x)} \! \ket{1} \! \ket{1} \!.
	\end{aligned}\mspace{-24mu}
\end{equation}
The $S$ gate linearly (but non-unitarily) maps ${\ket{0} \mapsto \ket{0}}$ and ${\ket{1} \mapsto \ket{0} + \ket{1}}$; thus the state after the $S$ operation on the final bit is
\begin{equation}
	\begin{aligned}[b]
		\ket{\Phi_3(x)} ={}&	
			\ket{x} \ket{0^m} \Bigl( \tfrac{1}{2}\ket{00} + \tfrac{1}{2}\ket{10} + \delta\sur{0}_x\ket{01} + \delta\sur{1}_x\ket{11} \Bigr)\ket{0} + \ket{\psi'_x} \ket{0}
		\\[-0.25ex]&
				+\delta\sur{L(x)}_x \ket{x}\ket{0^m} \ket{L(x)} \ket{1} \ket{1}.
		\\[0.5ex] ={}& \ket{\psi_x}\ket{0} 
				+\delta\sur{L(x)}_x \ket{x}\ket{0^m} \ket{L(x)} \ket{1} \ket{1}.
	\end{aligned}
\end{equation}
Conditioned on the final bit being $\ket{0}$, we coherently perform $U_n\herm$, and flip the second-to-last qubit conditioned on the final bit being $\ket{1}$.
This produces the state
\begin{equation}
	\begin{aligned}[b]
		\ket{\Phi_4(x)} ={}&	\ket{x}\ket{0^m}\ket{000} + \delta\sur{L(x)}_x \ket{x}\ket{0^m} \ket{L(x)} \ket{0}
		\ket{1};
	\end{aligned}
\end{equation}
removing the ancilla qubits, this yields a final state of
\begin{equation}
	\label{eqn:final-Wn-LWPP}
	\begin{aligned}[b]
		\ket{\varphi_x} ={}&	\ket{x}\Bigl(\ket{00} + \delta\sur{L(x)}_x\ket{L(x)}\ket{1}\Bigr).
	\end{aligned}
\end{equation}
A single non-unitary operation thus suffices to uncompute all components of the output, except for the $\ket{00}$ component and a component which indicates whether $x \in L$.

\subsection{Exact invertible algorithms}

If $(\mathbf N_0, \mathbf N_1)$ are dual \LWPP\ machines, there is a poly-time computable and length-dependent function $h$, for which $\delta\sur{L(x)}_x = \tfrac{1}{2^m} h(x)$.
The ability to compute $h(x)$ exactly and in polynomial time allows us to suppress the amplitude of the $\ket{00}$ component precisely, by interfering it with a (non-unitary) contribution from the component which indicates $L(x)$.
This allows us to easily prove that $L \in \GLP[\C]$, and that furthermore $L \in \EQP_{\GL}$ in the case that $L \in \LPWPP$.

\subsubsection{Polytime-specifiable invertible circuits for $L \in \LWPP$}

Figure~\ref{fig:decideLWPP} presents a polytime uniform (and polytime-specifiable) invertible circuit to exactly decide a language $L \in \LWPP$, in the case that $(\mathbf N_0, \mathbf N_1)$ are dual \LWPP\ machines.
In addition to the subroutine $W_n$ presented in Figure~\ref{fig:skew}, this circuit uses one gate which depends on the input size:
\begin{equation}
  A_n = \text{\small$\begin{bmatrix} h(1^n) & 0\;\; \\ 0 & 1\;\; \end{bmatrix}$}.
\end{equation}
By the choice of $h$, this is a family of invertible single-bit operations whose coefficients are computable in time $O(\poly n)$; it is therefore poly-time specifiable in the sense of Definition~\ref{def:UnitaryP}.
The algorithm also involves the two-bit invertible gate
\begin{equation}
	\label{eqn:twoBitDestrOptor}
  D = \text{\small$\begin{bmatrix}
				1	& \phantom{\!\!\!-}0 	& \!\!						-1	&		\!\!			-1	\\
				0	& \phantom{\!\!\!-}1	& \phantom{\!\!\!-}0	& \phantom{\!\!\!-}0	\\
				0	&	\phantom{\!\!\!-}0	&	\phantom{\!\!\!-}1	&	\phantom{\!\!\!-}0	\\
				0	&	\phantom{\!\!\!-}0	&	\phantom{\!\!\!-}0	&	\phantom{\!\!\!-}1
			\end{bmatrix}$},
\end{equation}
which linearly maps $\ket{0x} \mapsto \ket{0x}$ and $\ket{1x} \mapsto \ket{1x} - \ket{00}$; and the invertible operator $B = \diag(\tfrac{1}{2},1)$ described in Eqn.~\eqref{eqn:probabilityBoostOptor}.
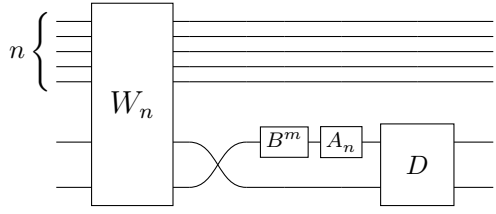
\begin{figure}
\begin{center}
\hspace*{-2.3em}
\begin{tikzpicture}
		\def\J{1,2,3,4,5}
		\foreach \j in \J {%
			\ifnum\j=1
				\coordinate (x\j-0) at (0,0);
			\else
				\coordinate (x\j-0) at ($(\prev) + (0,-0.2)$);
			\fi
			\xdef\prev{x\j-0}
		}

		\coordinate (c-0) at ($(\prev) + (0,-0.8)$);
		\xdef\prev{c-0}

		\coordinate (s-0) at ($(\prev) + (0,-0.6)$);
		\xdef\prev{s-0}
				
		\xdef\u{0}
		\foreach \t/\dt in {1/1,2/0.75,3/0.75,4/0.5,5/0.75,6/1,7/1}{
			\foreach \l in {x} {%
				\foreach \j in \J {%
					\coordinate (\l\j-\t) at ($(\l\j-\u) + (\dt,0)$);
					\draw [black] (\l\j-\u) -- (\l\j-\t);
				}
			}
			\foreach \l in {c,s} {%
				\coordinate (\l-\t) at ($(\l-\u) + (\dt,0)$);
			}
			\ifnum\t=3
				\coordinate (s-\u-5) at ($(s-\u)!0.5!(s-\t)$);
				\coordinate (c-\u-5) at ($(c-\u)!0.5!(c-\t)$);
				\draw (c-\u) .. controls (c-\u-5) and (s-\u-5) .. (s-\t);
				\draw (s-\u) .. controls (s-\u-5) and (c-\u-5) .. (c-\t);
				\coordinate (temp) at (c-\t);
				\coordinate (c-\t) at (s-\t);
				\coordinate (s-\t) at (temp);
			\else
				\draw [black] (c-\u) -- (c-\t);
				\draw [black] (s-\u) -- (s-\t);
			\fi
			\xdef\u{\t}
		}
			
		\node (x1-1) at (x1-1) {};
		\node (s-1) at (s-1) {};
		\node (W) at ($(x1-1)!0.5!(s-1)$) {\large$ W_n $};
		\node [draw,fill=white,fit=(x1-1)(s-1)(W)] {};
		\node (W) at (W) {\large$ W_n $};

		
		\node [draw,fill=white,inner sep=2pt] at (s-4) {\footnotesize$B^m_{\phantom\cdot}$};

		\node [draw,fill=white,inner sep=2pt] at (s-5) {\footnotesize$A_n^{\phantom\cdot}$};
		
		\node (c-6) at (c-6) {};
		\node (s-6) at (s-6) {};
		\node (D) at ($(c-6)!0.5!(s-6)$) {$ D_n $};
		\node (D') [draw,fill=white,fit=(c-6)(s-6)(D)] {};
		\node (D) at (D) {$ D $};

		\node (x) at ($(x1-0)!0.5!(x5-0)$) [anchor=east] {$n \left\{\mathclap{\begin{matrix} \\[4ex] \end{matrix}}\right.$\!\!\!\:};
\end{tikzpicture}
\end{center}
\caption{%
	Diagram for a circuit consisting of invertible gates to exactly decide a language $L \in \LWPP$.
	The invertible circuit $W_n$ is as depicted in Figure~\ref{fig:skew}, and depends on nondeterministic Turing machines representing $\LWPP$ algorithms for $L$ and $\overline L$.
	The gate $B$ is the same as in Eqn.~\eqref{eqn:probabilityBoostOptor}, and $D$ is given by Eqn.~\eqref{eqn:twoBitDestrOptor}.
	The gate $A_n$ acts on a single bit, and is polytime-specifiable from the input size $n$: in case we have $L \in \LPWPP$, we have $A_n = G^{O(\poly n)}$\!, for some single-qubit gate $G$ which is constant in the input size.
	On input $|x\rangle|0\rangle|0\rangle$, the output of this circuit is proportional to $|x\rangle|1\rangle|L(x)\rangle$.
}
\label{fig:decideLWPP}
\end{figure}%

Let $\ket{\varphi'_x}$ represent the result of swapping the final two bits of the state $\ket{\varphi_x}$ given in Eqn.~\eqref{eqn:final-Wn-LWPP}.
The effect of performing the circuit of Figure~\ref{fig:decideLWPP} on the input state $\ket{x}\ket{0}\ket{0}$ is then to produce the state
\begin{equation}
\begin{aligned}[b]
  \ket{\Phi_{L,x}} &= \Bigl(\idop\sox{n} \ox \bigl[ D (A_n B^m \ox \idop)\bigr] \Bigr) \ket{\varphi'_x}
  \\
  &= \ket{x} \ox \Bigl[ D \bigl(A_n B^m \ox \idop\bigr) \Bigl( \ket{0}\ket{0} + \delta_x\sur{L(x)} \ket{1}\ket{L(x)} \Bigr)\Bigr].
\end{aligned}
\end{equation}
The operator $A_n B^m \ox \idop$ multiplies the amplitude of $\ket{0}\ket{0}$ by $h(1^n)/2^m = h(x)/2^m$, leaving the component $\ket{1}\ket{L(x)}$ unaffected.
By construction, we have $\delta\sur{L(x)}_x = h(x)/2^m$, so i
t follows that
\begin{equation}
\begin{aligned}[b]
  \ket{\Phi_{L,x}} &= \ket{x} \ox \frac{h(x)}{2^m} \Bigl[ D \Bigl( \ket{0}\ket{0} + \ket{1}\ket{L(x)} \Bigr)\Bigr]
  = \frac{h(x)}{2^m} \ket{x} \ket{1}\ket{L(x)} .
\end{aligned}
\end{equation}
Thus the value of the final bit of the output is $\ket{L(x)}$, with certainty.
With Proposition~\ref{prop:exactInvertibleUpperBound}, we then have:
\begin{theorem}\label{thm:GLP}
  $\LWPP = \GLP[\C]$.
\end{theorem}

\subsubsection{Circuits for $L \in \LPWPP$ with finite invertible gate-sets}

In the case that $L \in \LPWPP$, the function $h(x)$ for which $\delta\sur{L(x)} = h(x)/2^m$ may be taken to be a power of some constant, $
  h(x) = M^{t(x)}	$ 
  for some $M\ge1$ and poly-time computable $t(x) \in O(\poly |x|)$.
From this it follows that 
\begin{equation}
	A_n = G^{t(x)},
	\qquad\text{%
		where $G = \diag(M,1)$.
	}
\end{equation}
Thus the $A_n$ gate in Figure~\ref{fig:decideLWPP} may be simulated by polynomially many $G$ gates.
The analysis of the preceding Section then suffices to show that $L \in \LPWPP$ may be exactly decided by a circuit over the gate-set $\bigl\{X,\mbox{\small CNOT},\mbox{\small TOFFOLI},H,S,B,G,D\bigr\}$
.
Together with Proposition~\ref{prop:exactInvertibleUpperBound}, this proves:
\begin{theorem}\label{thm:EQPGL}
  $\LPWPP = \EQP_{\GL}$.
\end{theorem}

\section{Remarks}

The results of this article are summarised in Figure~\ref{fig:summary}.
We have demonstrated that the classes \LWPP\ and $\Delta\Ceq\P$, which may seem to be of mainly technical interest in gap-definable complexity, can be defined quite naturally in terms similar to quantum circuit families.
Specifically, they may be characterised in terms of acyclic directed networks of invertible tensors.
The same perspective also motivates a gap-definable class \LPWPP\ which provides a tighter upper bound for \EQP.

Using techniques similar to that of Theorem~\ref{thm:GLP}, one can describe a similar characterisation of the gap-definable class \WPP, in terms of a family of circuit with input-dependent gates in place of the length-dependent gate $A_n$.
These may be considered to represent a sort of input-dependent ``just-in-time'' circuit, in which an invertible-gate circuit which is computed depending on an input $x$, and then immediately used as a subroutine to solve a decision or promise problem involving $x$.
Such a family of circuits would not satisfy any uniformity condition in the sense that ``uniformity'' is usually understood, and would represent a departure from the usual approach to computational complexity; nevertheless, this is a way in which \WPP\ may also be characterised in terms of tensor networks.
The problems solvable by unitary ``just-in-time'' circuit families of this sort would then be bounded by \WPP.
(It seems likely that the algorithm of Mosca and Zalka~\cite{MZ-2003} for \mbox{DISCRETE-LOG} is an example of such a just-in-time quantum algorithm: this would depend on a careful examination of how to exactly represent the state preparations involved through algebraic coefficients.)

Other gap-definable classes can also be described using tensor networks (albeit over fields other than $\C$), also representing models of indeterministic, ``quasi-quantum'' computation~\cite{Beaudrap-2014}.
Perhaps all of counting complexity can be recast in terms of tensor-like structures over semirings, in a manner not unlike Valiant's matchgates~\cite{Valiant-2005}, but with the added intuition that the networks represent modes of indeterministic computation.

Our analysis suggests that the technical aspects of the definitions of \LWPP\ and \WPP\ (and analogously, \AWPP) merely serve to capture indeterministic computations involving algebraic numbers, represented using natural numbers.
Viewed in this light, it is not obvious that those definitions could be made simpler while still being presented in terms of gap-functions, rather than efficiently specified tensor networks.
It also demonstrates that the existing upper bounds on quantum complexity classes make no use of the unitarity of the transformations involved.
It is not obvious how a condition such as unitarity would be represented by gap-functions --- except in the same technical manner as in the traditional definitions of \LWPP, \WPP, and \AWPP.

\begin{figure}[t]
\begin{minipage}[b]{0.62\textwidth}
\small
\begin{tikzpicture}[yscale=0.75, xscale=1.5, every node/.style={draw=none, inner sep=0pt, outer sep=3pt, rectangle}]
	\node (UnitaryP) at (3,3) {\UnitaryP[\C]};
	\node	(P) at (-1,0) {\P};
	\node (EQP) at (2,2) {\EQP};
  \node (ZQP) at (2,5) {\ZQP};
  \node (BQP) at (2,7) {\BQP};
  \node	(UP)	at (-1,2) {\UP};
  \node (SPP)	at (-1,3) {\SPP};
	\node (LPWPP) at (-1,4) {\LPWPP};
	\node [anchor=north] (EQPGL) at (LPWPP.north -| 0,0) {$\EQP_{\GL}$};
	\coordinate (EQPGLw') at (LPWPP.east -| EQPGL.west);
	\draw ($(LPWPP.east) + (0,1.5pt)$) -- ($(EQPGLw') + (0,1.5pt)$);
	\draw ($(LPWPP.east) + (0,-1.5pt)$) -- ($(EQPGLw') + (0,-1.5pt)$);
  \node (LWPP) at (-1,6) {\LWPP};
	\node [anchor=north] (GLP) at (LWPP.north  -| 0,0) {$\GLP[\C]$};
	\coordinate (GLPw') at (LWPP.east -| GLP.west);
	\draw ($(LWPP.east) + (0,1.5pt)$) -- ($(GLPw') + (0,1.5pt)$);
	\draw ($(LWPP.east) + (0,-1.5pt)$) -- ($(GLPw') + (0,-1.5pt)$);
  \node (DCeqP) at (-1,9) {$\Delta \Ceq\P$};
	\node [anchor=north] (ZQPGL) at (DCeqP.north -| 0,0) {$\ZQP_{\!\!\:\GL}$};
	\coordinate (ZQPGLw') at (DCeqP.east -| ZQPGL.west);
	\draw ($(DCeqP.east) + (0,1.5pt)$) -- ($(ZQPGLw') + (0,1.5pt)$);
	\draw ($(DCeqP.east) + (0,-1.5pt)$) -- ($(ZQPGLw') + (0,-1.5pt)$);
  \node (WPP) at (-1,7) {\WPP};
  \node (coCeqP) at (-1,11) {$\co\Ceq\P$};
  \node (NQP) at (coCeqP -| 0,0) {$\NQP$};
	\coordinate (NQPw') at (coCeqP.east -| NQP.west);
	\draw ($(coCeqP.east) + (0,1.5pt)$) -- ($(NQPw') + (0,1.5pt)$);
	\draw ($(coCeqP.east) + (0,-1.5pt)$) -- ($(NQPw') + (0,-1.5pt)$);
  \node (AWPP) at (1,9) {\AWPP};
	\node (PP) at (1,11) {$\PP \mathrlap{{}=\BQP_{\GL} \!\!\:= \Post\BQP}$};
	\draw (EQP) -- (ZQP) -- (BQP);
	\draw (P) -- (UP) -- (SPP) -- (LPWPP) -- (LWPP) -- (WPP) -- (DCeqP) -- (coCeqP); 
	\draw (P) -- (EQP);
	\draw [shorten >=-4pt] (EQP) -- (UnitaryP);
	\draw (UnitaryP) -- (BQP);
	\draw (EQP) -- (EQPGL);
	\draw [shorten <=4pt] (UnitaryP) -- (GLP);
	\draw (WPP) -- (AWPP);
	\draw (ZQP) -- (ZQPGL);
	\draw (ZQP) -- (BQP);
	\draw (DCeqP) -- (PP);
	\draw (BQP) -- (AWPP);
	\draw (AWPP) -- (PP);
	\begin{pgfonlayer}{background}
		\node [fill=gray!20!white, fit=(LPWPP)(SPP)(AWPP)(PP), inner xsep=20pt, inner ysep=-26pt] {};
	\end{pgfonlayer}
\end{tikzpicture}
~
\vspace*{9ex}
\end{minipage}
\hfill
\begin{minipage}[b]{0.3675\textwidth}
\caption{%
	Elaboration of Figure~\ref{fig:containments}, including the new results of this article.
	The quasi-quantum classes $\EQP_{\GL}$, $\GLP[\C]$, and $\ZQP_{\GL}$ are presented in the second-leftmost column, and are equal to the corresp\-onding gap-definable classes to their left (by Theorems~\ref{thm:ZQPGL}, \ref{thm:GLP}, and~\ref{thm:EQPGL}).
	Thus, the best known upper bounds for the quantum classes on the right (including ${\EQP \subset \LPWPP}$ and ${\UnitaryP[\C] \subset \LWPP}$, the newly presented bounds from Propositions~\ref{prop:LPWPPunitaryBound} and~\ref{prop:UnitaryPbound}) follow from the generalisation from unitary gates to invertible gates.
	One result not shown here for reasons of space is that the class $\Post\EQP$ of problems exactly solvable by postselected quantum algorithms is equal to $\Delta\Ceq\P$ (Theorem~\ref{thm:PostEQP}).
	Thus, despite the power of postselection in the bounded-error case, and the equivalence of postselection to invertible non-unitary gates in that case, 
	we can only show $\EQP_{\GL} \subset \Post\EQP \subset \Post\BQP$, with neither containment known to hold with equality.
}
\label{fig:summary}  
\end{minipage}
\vspace*{-5.5ex}
\end{figure}
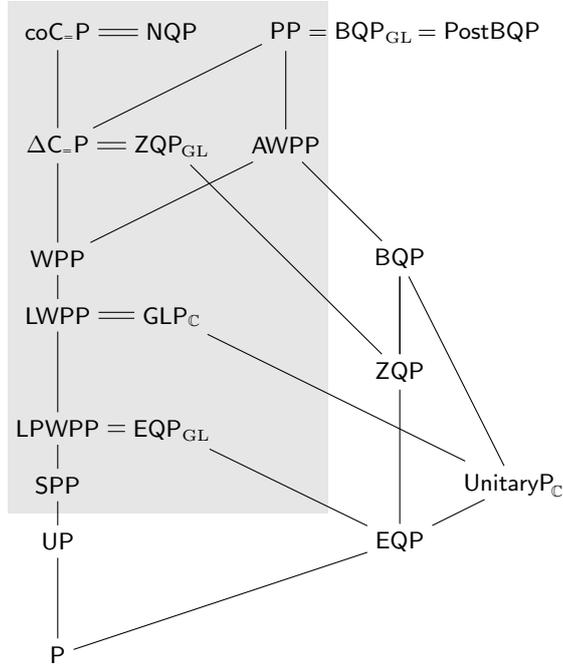

The question remains as to whether any of the bounds $\EQP \subset \LPWPP$, $\UnitaryP[\C] \subset \LWPP$, or $\ZQP \subset \Ceq\P \cap \co\Ceq\P$, may be tightened or shown to be strict.
The classes $\EQP_{\GL}$ and $\GLP[\C]$ should be expected to be much more powerful than their unitary counterparts; to reiterate the argument from the introduction, by Valiant--Vazirani~\cite{VV-1986}, if either $\UP \subset \EQP$ or $\UP \subset \UnitaryP[\C]$ it would follow that $\NP \subset \BQP$.
One could perhaps consider intermediary classes by requiring the singular values of the gates in the circuits to be nearly $1$, or similar restrictions to near-unitarity.

It is not yet known whether any of the containments $\P \subset \EQP \subset \UnitaryP[\C]$ are strict.
However, we at least see that exact \emph{quantum-like} classes can solve problems which are not expected to be in $\P$; thus 
Thus the constraint of exact decision does not in itself prevent a fruitful theory of indeterministic algorithms.
Furthermore, as non-trivial techniques will be required to simulate unitary circuits with polynomial-time coefficients using circuits built from a finite gate-set, it seems likely that $\EQP \ne \UnitaryP[\C]$.

We ask whether $\SPP$ has an interpretation as a quantum-like class in the manner of $\LWPP$ and $\Delta\Ceq\P$, and whether \SPP\ is strictly contained in $\EQP_{\GL} = \LPWPP$.
A proof that $\SPP = \LPWPP$ could be taken as evidence that $\EQP = \P$; we consider either of these possible equalities to be worthy topics of investigation.
Finally, we ask whether results similar to these also hold for other quasi-quantum models of computation, such as transformation by affine operators (which we allude to on page~\pageref{discn:affine}).
A result implicit in Ref.~\cite{BLBH-2015} is that the problems which may be solved with bounded error by polynomial-time affine circuit families over a finite gate set is the entire class \AWPP; we expect that the class of problems exactly solvable by such circuits is again \LPWPP, and that extending to polytime-specifiable gate sets again recovers \LWPP.

\vspace*{-1ex}
\paragraph{Acknowledgements.}

This work was performed in part at the CWI, with support from a Vidi grant from the Netherlands Organisation for Scientific Research (NWO) and the European Commission project QALGO.
A first draft of this work was completed while a guest of Jane Biggar, during a break from professional academia.
I would like to acknowledge support from the UK Quantum Technology Hub project NQIT.

\appendix
\section{On quantum algorithms with ``infinite'' gate sets}
\label{apx:manifesto}

It might be that $\EQP$ contains no interesting problems (or even no problems at all) outside of $\P$.
The results of this article show that exactness alone is not the cause.
The other significant constraints on the study of exact quantum algorithms is the necessary restriction to unitary operations, and the restricting to a fixed finite gate set, both of which were originally motivated by the original definition of \EQP\ in terms of quantum Turing machines~\cite{BV-1997}.

As we note on page~\pageref{discn:factoring} (and as observed by Nishimura and Ozawa~\cite{NO-2002}), a prime-factor version of the integer factoring problem is in \ZQP, so it is not obvious that these constraints are too strict.
However, in the zero-error case, the freedom to fail with constant probability allows us to exploit standard approximation results such as the Solovay--Kitaev Theorem~\cite{KSV-2002,DN2006} to simulate quantum Fourier transforms (QFTs) of arbitrary order.
Such freedom to fail is not available to exact quantum algorithms; nor can QFTs of arbitrary order (which involve algebraic number fields of arbitrarily large degree) be decomposed into a finite gate set, which can only yield coefficients from a fixed algebraic number field.
In the case of integer factorisation, this is not the only obstacle for exact quantum algorithms.
However, this does seem indicative of the central problem for exact quantum complexity: that quantum-versus-classical speedups require algorithms involving ``large'' amounts of destructive interference, and that arranging for the destructive interference to be \emph{total} (in the physicist's sense) is difficult to arrange with only a finite gate set.

The alternative is to allow infinite gate-sets --- or perhaps more appropriately, to allow the basis $\mathcal G_n$ available to a circuit $C_n$ in a uniform circuit family $\{ C_n \}_{n \ge 1}$ not to be bounded in size by a constant.
This is not \emph{a priori} unreasonable from the standpoint of computational complexity: classical $\mathsf{AC}^k$ circuits (with poly-log depth and gates with unbounded fan-in) also have length-dependent gate-sets --- albeit ones which may be exactly decomposed by virtue of the existence of a finite universal gate sets for boolean logic.
The absence of such finite universal gate sets for quantum computation need not prevent us from considering length-dependent gate-sets; indeed, one might suppose that this absence \emph{forces} the issue of such gate-sets on us, as we cannot rule them out without loss of generality in the case of exact algorithms.

A consequence of considering such gate-sets is a break from the exact correspondence between quantum Turing machines and quantum circuit families.
However, insisting on such a correspondence introduces a distinction in quantum complexity between ``algorithms'', and ``meta-algorithms''~\cite[Appendix~C]{Beaudrap-2014} consisting of efficiently computable \emph{descriptions} of other computations.
The standard theory of universal Turing machines prevents such a distinction in classical computational complexity; the absence of an exact (and efficient\footnote{%
	There do exist exact universal quantum Turing machines, if one is willing to dispense with efficient simulation.
	For example, any deterministic universal Turing machine suffices.
})
universal quantum Turing machine is what makes an algorithm versus meta-algorithm distinction possible in the quantum regime.
However, it is not clear that it is productive or meaningful to make such a distinction, especially as unitary circuit families and adiabatic quantum computations are both descriptions of computations whose parameters are computed by a Turing machine.

If one were to consider circuit families, which involve gate sets which scale with the circuit size, one must consider how to do so without trivialising the theory of quantum complexity, for either exact algorithms or for bounded-error algorithms.
It seems plausible that restricting them to be polynomial-time specifiable, so that the entire tensor network of the each circuit $C_n$ can be specified in time $O(\poly n)$, is a reasonable first step.
In order to avoid subsuming a non-trivial amount of work into any one gate --- \eg~polylogarithmic factors, due to sophisticated unitary operators acting on $O(\log n)$-qubits --- one may impose further restrictions.
Given that we require each gate in the circuit to be explicitly expressible in polynomial time, it seems likely that standard exact circuit decomposition techniques~\cite{NC-2000} would allow us to restrict to gates acting on $O(1)$ qubits without loss of generality.
Furthermore, to represent the cost of calibration of (abstract) devices which control the unitary evolution of the computational state, and also to account for the work performed by a straightforward evaluation of amplitudes through the gap-function of an NTM, it seems reasonable to impose a non-trivial cost on each gate, which depends on the time required to compute its description.
(For a gate-set of finite size, this cost would be bounded by a constant, recovering the usual notion of circuit size as the cost of the circuit.)
Ref.~\cite[\S3.3\,C]{Beaudrap-2014} sketches how circuits involving such gates may be simulated with bounded error by constant-size gates; and Proposition~\ref{prop:UnitaryPbound} of this article sketches how the complexity of algorithms using such gate-sets can be bounded by counting classes which are well-known as upper bounds for exact quantum complexity.
Thus, such an approach seems likely to leave the existing theory of bounded-error quantum computation largely unchanged, while providing more freedom in the study of exact quantum algorithms.

It seems plausible that there would be problems of interest in $\UnitaryP[\C]$ which are not expected to be in \P, or even \BPP.
The usual approach to decomposing the QFT over the integers modulo $2^n$ taught in many introductory courses on quantum algorithms~\cite{Coppersmith-1994}, is an example of a uniform circuit family over a polytime-specifiable gate-set.
It seems plausible that, by some modification of the approach of Mosca and Zalka~\cite{MZ-2003}, some version of \mbox{DISCRETE-LOG} might be contained in the class \UnitaryP[\C] of problems which are exactly decidable by such unitary circuits (though this may involve technical arguments on the number of $O(n)$-bit primes for which \mbox{DISCRETE-LOG} is expected to be hard).
As we show in the main text, \UnitaryP[\C] is bounded above both by \LWPP\ (until recently the best known bound on \EQP) and \BQP, which we consider evidence that this approach to quantum computational complexity is not extravagant.
We therefore propose the study of the class \UnitaryP[\C] (or some similar notion of exact quantum complexity) as a way to approach the subject of exact quantum algorithms.

\end{document}